\documentclass[runningheads]{llncs}
\usepackage[T1]{fontenc}
\usepackage{graphicx}
\usepackage{setspace} % 行距控制

\makeatletter
\renewcommand{\normalsize}{%
   \@setfontsize\normalsize{11}{13.2}% % 11pt字号，13.2pt基础行距(1.2倍)
   \abovedisplayskip 11\p@ \@plus3\p@ \@minus6\p@
   \abovedisplayshortskip \z@ \@plus3\p@
   \belowdisplayshortskip 6.5\p@ \@plus3.5\p@ \@minus3\p@
   \belowdisplayskip \abovedisplayskip
   \let\@listi\@listI}
\makeatother

\usepackage{etoolbox} % 提供环境修改工具
\patchcmd{\abstract}{\small}{\normalsize}{}{} % 将\small改为\normalsize
\usepackage[margin=1in]{geometry}
\usepackage{amsmath,amsfonts,amssymb}
\usepackage{booktabs}
\usepackage{pifont}
\usepackage[ruled,linesnumbered,vlined]{algorithm2e}
\usepackage{float}
\usepackage{mathrsfs}
\usepackage{multirow,multicol}
\usepackage{subfigure}
\usepackage{enumitem}
\usepackage[all]{xy}
\usepackage{xcolor}
\usepackage{bbm}
\usepackage{makecell}
\usepackage{pdfpages}
\usepackage{hyperref}
\usepackage{times}
\usepackage{soul}
\usepackage{url}
\usepackage{newfloat}
\usepackage{listings}
\usepackage{tikz}
\usepackage{appendix}
\usepackage{lineno}
\usepackage{marvosym}

\def\nil{\mathbf{nil}}

\renewcommand\ge{\geqslant}

\SetAlgorithmName{Mechanism}

\begin{document}
\title{Probabilistic Mechanism Design in Diffusion Auctions}

    \author{Xinlun Zhang\inst{1} \and
        Zhechen Li\inst{2} \and
        Yongzhi Cao\inst{1}\Letter \and
        Yu Huang\inst{3} \and
        Hanpin Wang\inst{1}
        }
	

    \authorrunning{X. Zhang et al.}
    \institute{Key Laboratory of High Confidence Software Technologies, School of Computer Science, Peking University, China \\
    \email{\{ec\_lol, caoyz, whpxhy\}@pku.edu.cn} 
    \and
    Ant Group, Beijing, China \\
    \email{lizhechen@pku.edu.cn} 
    \and National Engineering Research Center for Software Engineering, Peking University, China \\
    \email{hy@pku.edu.cn}}
    %\email{xialirong@gmail.com} \and
    %School of Computer Science and Cyber Engineering, Guangzhou University, China
    %}

\maketitle

\begin{abstract}
    A diffusion auction refers to a selling process conducted over a social network, where each participant submits a bid and may invite other potential buyers to join the auction. Although various mechanisms have been proposed, none of them can simultaneously achieve incentive compatibility, non-negative revenue, and approximate efficiency with a constant approximation bound. In this paper, we propose the Probabilistic Diffusion Mechanism (PDM), a novel mechanism tailored for path graphs, which satisfies all three desired properties. We further extend PDM to general network structures through a map $f$, resulting in the $f$-PDM mechanism, which preserves the key properties of the original design. Beyond these, when $f$ satisfies properties such as breadth-first order, $f$-PDM also ensures Sybil-proofness and provides approximate revenue. Furthermore, to address buyer collusion, we introduce a modified version of the mechanism that balances collusion-proofness with revenue approximation. Finally, we extend the design to multi-unit diffusion auctions---a more challenging setting---and propose a simple yet effective mechanism, Multi-Unit PDM (MUPDM), that achieves approximate efficiency while maintaining IC. Moreover, we design Sybil-Proof MUPDM (SP-MUPDM) to resist Sybil attacks in the multi-item scenario.
    \keywords{Auction \and Mechanism Design \and Social Network \and Collusion-proofness \and Sybil-proofness}
\end{abstract}

\section{Introduction} 
% Diffusion auction, proposed by \cite {idm}, is an emerging field in auctions. Unlike traditional auctions, buyers in diffusion auctions not only need to bid, but also invite others to participate in the auction. On the one hand, more buyers bring higher social welfare and seller income. On the other hand, how to motivate buyers to invite others becomes a problem.
	The diffusion auction, introduced by Li et al. ~\cite{idm}, establishes a new paradigm in auction theory by operating over a social network. In this setting, each buyer is required not only to submit a bid, but also to recruit other participants through their social connections. Unlike conventional auctions---where the mechanism input consists solely of bids---a buyer’s report in a diffusion auction is a tuple that includes both a bid and a set of invited neighbors. Despite this structural difference, the design of diffusion mechanisms follows the same approach used in traditional auction theory, where mechanisms are evaluated against a set of desirable properties. Key properties include individual rationality (IR; truthful reporting guarantees non-negative utility), incentive compatibility (IC; truth-telling is a dominant strategy), weak budget balance (WBB; the seller’s revenue is non-negative), and efficiency (the item should be allocated to the highest-value buyer).
	%Diffusion auction, introduced by Li et al.~\cite{idm} is a new paradigm in auction theory. As the name suggests, it is conducted over a {\em social network}. Consequently, each buyer in any certain diffusion auction is required to recruit other participants through their social connections. Unlike conventional auction mechanisms, which only collects the bids of buyers as input, the profile that a buyer should submit to a diffusion auction mechanism is a tuple consisting of a bid and a group of participants. Although the inputs are different, the mechanism design of diffusion auction follows the same paradigm as the conventional auction, i.e., to evaluate mechanism through {\em axiomatic properties}. Popular properties includes individual rationality (truthfully reports leads to non-negative utility), incentive compatibility (which requires telling the truth to be a dominant strategy), and efficiency (the highest bidder wins). 
	% In the past few years, many diffusion auction mechanisms have been proposed. All of them were evaluated by the properties.
	% 想到哪写到哪，写得不太好 你自己改改。
	
	Due to the difference on inputs, the requirement of these properties are somehow ``enhanced'' for diffusion auctions. For example, for conventional auctions, the incentive compatibility only requires the buyers to report their bids truthfully. In diffusion auctions, it requires buyers to truthfully report not only their valuations but also their neighbors. However, buyers often lack the motivation to invite others, as increased competition from additional participants may reduce their own utility. A central challenge in diffusion auction design is therefore to create mechanisms that motivate buyers to actively diffuse auction information.
	
	Beyond the properties mentioned above, diffusion auctions should also address strategic behaviors such as Sybil attacks and collusion. A Sybil attack occurs when a buyer creates fake identities (called Sybil identities) to gain an unfair advantage, while collusion involves a group of buyers (forming a cartel) coordinating to increase their collective utility. Currently, few mechanisms provide Sybil-proof guarantees, and virtually none offer collusion resistance. This highlights a significant gap in the design of secure diffusion auction mechanisms.
	% 与传统拍卖不同的是，扩散拍卖中的IC不仅要求买家诚实报告自己的估值，还要求买家诚实报告自己的邻居。然而，买家往往没有动机邀请他人，因为更多的买家带来了更多的竞争，可能会减少自己的效用。因此，扩散拍卖一个关键挑战是如何设计合理的机制。除了上述性质之外，扩散拍卖还需要考虑如何防止女巫攻击、合谋攻击等策略性行为。女巫攻击是指一个买家通过创建虚假身份（称为女巫身份）来获得不公平的优势，而合谋攻击是指多个买家（称为卡特尔）进行合谋来提升他们整体的效用。然而，很少机制满足防女巫攻击，甚至没有机制能够抵抗合谋，在设计安全机制方面的研究还有很大缺失。
	
	In other words, the following question is still largely open.
	\begin{center}
		\em How can we design mechanisms that satisfy a broader set of desirable properties?
	\end{center}
	
	In addition, extending from the design of single-item auction mechanisms, we generalize the framework to multi-item settings while preserving the original properties.

	% Li等人提出的扩散拍卖是拍卖理论中的一种新范式。正如它的名字所展现的，它是在社会网络基础上构建的拍卖。因此，扩散拍卖要求买家不仅提交出价，而且还通过其社会关系积极招募其他参与者，这弥补了传统拍卖中买家不知道拍卖的存在而没有参加所导致的社会福利和卖家收入的下降的不足。
	
	%与传统拍卖类似，机制设计是十分关键的一部分。在机制设计中，衡量
	
	%然而，买家往往没有动机邀请他人，因为更多的买家带来了更多的竞争，可能会减少自己的效用。因此，扩散拍卖一个关键挑战是设计有效的激励机制，鼓励买家如实报告自己的估值并在社交网络中传播拍卖信息，这被称为激励相容（IC）。
	
	%除了IC之外，其他传统拍卖的一些基本性质也是评价扩散拍卖机制的指标。如个体理性（IR，买家诚实报告估值的效用非负）、弱预算平衡（WBB，卖家收入非负）、有效的（出价最高的买家赢得物品）。然而，李等人表明，不存在同时满足上述4个性质的扩散拍卖机制。如表1所示，目前的机制主要牺牲了有效性来换取激励相容，也有少部分放弃激励相容来满足有效性。然而，目前还没有机制做到同时满足三个性质并且对其余一个性质有常数近似的保证。
	
	%扩散拍卖的设计不仅要考虑上述基本性质，还需要考虑如何防止女巫攻击、合谋攻击等策略性行为。女巫攻击是指一个买家通过创建虚假身份（称为女巫身份）来获得不公平的优势，而合谋攻击是指多个买家（称为卡特尔）进行合谋来提升他们整体的效用。然而，很少机制满足防女巫攻击，甚至没有机制能够抵抗合谋，在设计安全机制方面的研究还有很大缺失。
	
	% \begin{center}
		%     \em Is there any mechanism satisfying a, b, and c?
		% \end{center}
	
	%我们尝试设计满足更多性质的机制——无论是近似的社会福利和卖家收入，还是防女巫攻击、防合谋攻击等安全保障。
	%

	% Contribution: a, b, c
	
	\subsection{Related Work}
	For single-unit diffusion auctions, Li et al. showed that the VCG mechanism \cite{vcg1,vcg2,vcg3} guarantees IC but may violate weak budget balance. They then proposed the information diffusion mechanism (IDM) as a solution that achieves both IC and WBB \cite{idm}. Subsequently, Li et al. ~\cite{cdm1,cdm2} generalized IDM to a class of mechanisms, termed the critical diffusion mechanism (CDM), which retain the same properties as IDM. However, both IDM and CDM only reward buyers who play critical roles in diffusion. According to the small-world theorem ~\cite{smallworld}, the probability of being a critical node in a social network is extremely low, resulting in limited incentives for buyers to participate in the diffusion process. To solve this problem, the fair diffusion mechanism (FDM) \cite {fdm} and the network-based redistribution mechanism (NRM) \cite {nrm} were proposed.
	
	%在单物品扩散拍卖的机制设计中，李等人首先说明了VCG机制是IC但可能导致卖家亏损，进而提出了满足IC和WBB的IDM机制。
	%在满足基本性质的基础上，很多工作尝试设计满足额外性质的机制。Chen等人设计了防女巫攻击的机制the Sybil tax mechanism (STM) and the Sybil cluster mechanism (SCM)，Jia等人设计了满足差分隐私的机制recursive DPDM and layered DPDM，Jeong and Lee 提出的 the groupwise-pivotal referral (GPR) mechanism 机制则能够部分地防止合谋攻击。
	Building upon the basic properties, several follow‑up works have designed mechanisms that achieve additional desirable properties. For example, Chen et al. proposed Sybil‑resistant mechanisms---the Sybil tax mechanism (STM) and the Sybil cluster mechanism (SCM) ~\cite{sybil2}. Zhang et al. modified Myerson's mechanism \cite{myerson1981} and proposed the $k$‑partial winner of Myerson's mechanism to provide an approximation of the seller's revenue \cite{optimal}. However, this mechanism requires the distribution of buyers' valuations to be common knowledge. Jia et al. introduced mechanisms satisfying differential privacy, namely the recursive DPDM and the layered DPDM ~\cite{dpdm2,dpdm1}. Meanwhile, the groupwise‑pivotal referral (GPR) mechanism by Jeong and Lee provides partial protection against collusion ~\cite{gpr}.
	%When designing diffusion auction mechanisms, additional properties, such as Sybil-proofness, differential privacy, and collusion-proofness are considered.
	%In diffusion auctions, Sybil attacks occur when a buyer creates multiple identities to gain excessive benefits, as mechanisms may reward buyers for participating in diffusion. To address this issue, two Sybil-proof mechanisms, the Sybil tax mechanism (STM) and the Sybil cluster mechanism (SCM) ~\cite{sybil2}, were proposed.
	
	%In addition to being vulnerable to 
	%Beyond Sybil attacks, diffusion auctions face privacy leakage risks. To counter this issue, differential privacy ~\cite{dp} ensures similar inputs produce statistically similar outputs. Jia et al. ~\cite{dpdm2,dpdm1} proposed two differentially private mechanisms (recursive DPDM and layered DPDM), which are the only stochastic mechanisms in single-unit diffusion auctions to date. Although layered DPDM offers a lower bound on expected social welfare, this bound depends on the network structure and may result in no buyer obtaining the item.
	
	%Collusion-proofness ensures that multiple buyers cannot collude to gain more benefits, which has been studied in auctions. Diffusion auctions are also vulnerable to collusion, but no current mechanism effectively tackles this challenge. In IDM, multiple buyers may collude and become one single buyer, thereby becoming critical nodes and gaining benefits. Jeong and Lee \cite {gpr} proposed  the groupwise-pivotal referral (GPR) mechanism and demonstrated that it could partially prevent collusion. However, the mechanism is not IC.
	
	In addition to positive results, there are significant limitations in diffusion auction mechanisms. Li et al. ~\cite{ic} demonstrated that no diffusion auction mechanism can simultaneously satisfy incentive compatibility, weak budget balance, and efficiency (i.e., maximizing social welfare). Even worse, Chen et al. ~\cite {sybil2} proved that the social welfare and revenue of any (deterministic) Sybil-proof and weak budget balance mechanism could be arbitrarily low in the worst-case scenario, highlighting a critical drawback of existing approaches.
	%Most of the existing diffusion auction mechanisms give up efficiency to satisfy incentive compatibility and weak budget balance. However, their social welfare is zero in the worst-case scenario and mechanisms of approximate efficiency guarantee have been rarely studied.
	
	Multi-unit auctions naturally extend single-unit auctions. However, in multi-unit diffusion auctions, mechanism design becomes significantly more challenging, as each buyer can choose whether to invite others to control their payment and item allocation. Although the GIDM mechanism ~\cite{mul1} and the DNA-MU mechanism ~\cite{mul3} attempted to address multi-unit auctions, they were later shown to fail incentive compatibility ~\cite{mul4,mul6,mul7}.
	The LDM Tree mechanism ~\cite{mul7} guarantees IC by restricting auctions within each hierarchical level, while the PDA mechanism ~\cite{mul8} discussed fairness in multi-unit settings.
	The MUDAN mechanism ~\cite{mul6} ensures IC and weak efficiency through graph exploration, whose efficiency approximation ratio applies only to the buyers explored, rather than to the entire participants.
	
	Beyond the multi‑item setting, diffusion auction models have been extended to variants such as weighted networks \cite{cdm1,weight1,weight2,mul-weight}, auctions with budget constraints \cite{mul5}, double auctions \cite{double2}, two‑sided matching \cite{match2}, and reverse auctions \cite{reverse}. The framework has further found applications across diverse fields including cooperative game theory \cite{co}, housing markets \cite{housing1,housing2}, crowdsourcing \cite{crowdsourcing1,crowdsourcing2,crowdsourcing3}, and fog computing \cite{fog}. Moreover, formal verification of diffusion auctions has also been studied \cite{verify2,verify1}.
	
	%除了多物品模型之外，还包括带权网络、双重拍卖、双边匹配和逆向拍卖等扩散拍卖模型。。在应用方面，扩散拍卖模型还被用于合作博弈、住房市场、众包模型和雾计算等领域。此外，还有工作对扩散拍卖机制进行形式化验证。
	
	\subsection{Our contributions} 
	This paper proposes novel diffusion auction mechanisms that achieve stronger theoretical guarantees, including previously unrealized properties such as constant-bounded approximate efficiency and collusion-proofness. Since deterministic incentive compatible mechanisms cannot ensure approximate efficiency, as demonstrated in ~\cite{cdm2}, a natural idea is to introduce a stochastic mechanism. We first address the worst-case scenario of a path graph, which is critical for determining non-trivial approximate efficiency, and then extend our approach to general graphs and multi-unit auctions. Our contributions are summarized as follows.
	
	\begin{table*}[t]
		\centering
		%\begin{minipage}[t]{0.55\linewidth} %
		\caption{The properties of different diffusion auction mechanisms, where ``\ding{51}'' indicates that the row mechanism satisfies the column property, and ``\ding{55}'' indicates that the mechanism does not satisfy the property. Here, non-trivial efficiency and revenue require the expectation of social welfare and the seller's revenue to be strictly positive for any possible input. The \ding{51}\textsuperscript{$-$} indicates that DPDM satisfies non-trivial efficiency, but cannot provide any non-trivial constant lower bound. The \ding{51}\textsuperscript{$\dagger$} indicates that $f$-PDM can be modified to achieve collusion-proofness by accepting a trade-off in approximate revenue.}
		\label{tab:example}
		\begin{tabular}{l@{\quad}c@{\quad}c@{\quad}c@{\quad}c@{\quad}c@{\quad}cc}
			\toprule[1.5pt]
			& \textbf{IR} & \textbf{IC} & \textbf{WBB} & \textbf{Non-trivial efficiency} &  \textbf{SP} & \textbf{CP} & \textbf{Non-trivial revenue} \\
			\midrule[1pt]
			\quad IDM ~\cite{idm}   & \ding{51} & \ding{51} & \ding{51} & \ding{55} & \ding{55} & \ding{55} & \ding{55}\\
			\quad CDM ~\cite{cdm1}   & \ding{51} & \ding{51} & \ding{51} & \ding{55} & \ding{55} & \ding{55} & \ding{55}\\
			\quad FDM ~\cite{fdm}   & \ding{51} & \ding{51} & \ding{51} & \ding{55} & \ding{55} & \ding{55} & \ding{55}\\
			\quad NRM ~\cite{nrm}    & \ding{51} & \ding{51} & \ding{51} & \ding{55} & \ding{55} & \ding{55} & \ding{55}\\
			\quad STM ~\cite{sybil2}   & \ding{51} & \ding{51} & \ding{51} & \ding{55} & \ding{51} & \ding{55} & \ding{55}\\
			\quad SCM ~\cite{sybil2}   & \ding{51} & \ding{51} & \ding{51} & \ding{55} & \ding{51} & \ding{55} & \ding{55}\\
			\quad DPDM ~\cite{dpdm1}   & \ding{51} & \ding{51} & \ding{51} & \ding{51}\textsuperscript{$-$} & \ding{55} & \ding{55} & \ding{55}\\
			\quad GPR ~\cite{gpr}   & \ding{51} & \ding{55} & \ding{51} & \ding{51} & \ding{51} & \ding{55} & \ding{55}\\
			\quad $f$-PDM   & \ding{51} & \ding{51} & \ding{51} & \ding{51} (Thm. \ref{appeff}) & \ding{51} & \ding{51}\textsuperscript{$\dagger$} (Thm. \ref{collu}) & \ding{51} (Thm. \ref{apprev})\\
			\bottomrule[1.5pt]
		\end{tabular}
		%\end{minipage}
	\end{table*}
	
	\begin{enumerate}[label*=\arabic*)]
		%\item We introduce the Probabilistic Diffusion Mechanism (PDM), a stochastic mechanism for path graphs that guarantees individual rationality (IR), incentive compatibility (IC), weak budget balance (WBB), and Sybil-proofness (SP).
		
		\item We propose the Probabilistic Diffusion Mechanism (PDM), a diffusion auction mechanism tailored for path graphs, and extend it to arbitrary networks through a map $f$, thereby defining a family of mechanisms termed $f$-PDM. This family preserves the core properties of PDM---IR, IC, and WBB---while guaranteeing approximate efficiency with a constant approximation ratio. Moreover, with appropriately chosen $f$, $f$-PDM also achieves SP and delivers a non‑trivial approximation of seller revenue, as summarized in Table \ref{tab:example}.
		
		\item We next provide a formal definition of collusion in diffusion auctions and show that existing mechanisms are vulnerable to such behaviors. Building on this analysis, we further refine our mechanism to obtain collusion-proofness (CP) guarantees.
		
		\item Building upon PDM, we design two multi‑item diffusion auction mechanisms: MUPDM and SP‑MUPDM. The former not only satisfies IR, IC, and WBB, but also guarantees approximate efficiency; the latter sacrifices approximate efficiency in exchange for Sybil‑proofness.
		%我们在PDM的基础上设计了两个多物品扩散拍卖的机制：MUPDM和SP-MUPDM。前者不仅能够满足IR, IC, and WBB，还能够保证approximate efficiency；后者通过牺牲approximate efficiency来换取SP。
	\end{enumerate}
	
	The remainder of the paper is organized as follows. In Section \ref{sec2}, we introduce the model. Section \ref{sec4} presents our mechanism in single-unit diffusion auction and multi-unit settings are discussed in Section \ref{sec5}. We conclude the paper and discuss future research directions in Section \ref{sec6}.
	
	%\begin{enumerate}[label*=\arabic*)]
	%	\item We propose a novel stochastic mechanism named Probabilistic Diffusion Mechanism (PDM), which works when the network is a path graph. PDM ensures individually rationality (IR), incentive compatibility (IC), weak budget balance (WBB), Sybil-proof (SP), and approximate efficiency.
	
	%	\item We further generalize our mechanism to the graph by a map $f$ and design a class of mechanisms called $f$-PDM. $f$-PDM not only satisfies IC and WBB, but also ensures approximate efficiency with a constant bound. When $f$ is appropriately designed, $f$-PDM can also ensure SP and non-trivial approximate guarantee for the seller's income, which is shown in Table \ref{tab:example}.
	
	%	\item We model collusion in diffusion auctions and show the existing mechanisms cannot prevent collusion. We improve our mechanism to be collusion-proof (CP).
	%\end{enumerate}
	
	%The remainder of the paper is organized as follows. In Section \ref{sec2}, we establish the model. Our mechanism when the network is a path graph is shown in Section \ref{sec3}. In Section \ref{sec4}, we extend our mechanism to graphs. Finally, we conclude and discuss future work in Section \ref{sec5}.
	
	%------------------------ chap 2 ---------------------------
	\section{Preliminaries} \label{sec2}
	
	In a diffusion auction ~\cite{idm}, $N = \{1, 2, \ldots, n\}$ is a set of $n\ge 1$ buyers and $G=(N,E)$ is a {\em social network} upon it, where buyer $i$ knows buyer $j$ if and only if the directed edge from $i$ to $j$ is in the edge set $E\subseteq N \times N$. Further, the set of all buyers known by buyer $i$ is denoted by $r_i\subseteq N$. A seller $s$, who only knows some of the buyers $r_s \subseteq N$, is selling an item in $G$. Each buyer $i$ has a private valuation $v_i\in[0,v]$, where $v\in\mathbb{R}^+$ is a constant upper bound for the valuations of all buyers (without loss of generality, we assume that $v=1$ throughout the paper). For simplicity, let $t_i = (v_i, r_i)$ for all $i\in N$, called the {\em type} of buyer $i$.

	In the process of a diffusion auction, participation is restricted to buyers who have received an invitation.
	Initially, only buyers in $r_s$ are invited. Then each invited buyer $i$ reports a type $t'_i = (v'_i, r'_i)$, where $v'_i \in [0, 1]$ represents her bid for the item and $r'_i \subseteq r_i $ denotes the set of neighbors she wants to invite to the auction. The set of all possible types that buyer $i$ can report\footnote{If buyer $i$ is not invited, we have $t'_i = \nil$.} is denoted by $T_i$, i.e., 
	\begin{align*}
		T_i=[0,1]\times \mathcal{P}(r_i)\cup \{\nil\}, \text{ for all }i\in N,
	\end{align*}
	where $\mathcal{P}(r_i)$ denotes the power set of $r_i$. For simplicity, let $\boldsymbol{t}, \boldsymbol{t}' \in \boldsymbol{T}$ denote the vector of true types and reported types of all buyers, respectively, i.e.,
	$
	\boldsymbol{t}= (t_1, t_2, \ldots, t_n),\; \boldsymbol{t}'= (t'_1, t'_2, \ldots, t'_n),
	$
	where $\boldsymbol{T} = T_1\times T_2\times \cdots\times T_n$. Further, we use $\boldsymbol{t}_{-i}, \boldsymbol{t}'_{-i} \in \boldsymbol{T}_{-i}$ to denote the true types and the reported types of all buyers except $i$.
	
	Under the settings above, a (randomized) diffusion auction mechanism can be defined as follows, which can be regarded as an extension of the definition in ~\cite{idm}.
	
	\begin{definition}
		\label{def: rand-diffusion-auction}
		A randomized diffusion auction mechanism $M=(\boldsymbol{\pi}, \boldsymbol{p})$ consists of the following two components.
		\begin{enumerate}[label*=\arabic*)]
			\item An allocation rule $\boldsymbol{\pi}\colon \boldsymbol{T} \to [0,1]^n$, where for any input $\boldsymbol{t}\in \boldsymbol{T}$, the $i$-th element of $\boldsymbol{\pi}(\boldsymbol{t})$ represents the probability that the item is assigned to buyer $i$ (in which case we say buyer $i$ wins the item).
			\item A payment rule $\boldsymbol{p}\colon \boldsymbol{T} \rightarrow \mathbb{R}^{n\times n}$, where for any input $\boldsymbol{t}\in \boldsymbol{T}$, the $(i,j)$-th element of $\boldsymbol{p}(\boldsymbol{t})$ denotes the amount of money that buyer $i$ should pay to the seller when buyer $j$ wins the item.
		\end{enumerate}
	\end{definition}
	
	% \begin{definition}
		%     A randomized diffusion auction mechanism $M$ consists of two components:
		%     an allocation rule $\mathbf{\pi} = \pi _{i \in N}$ and a payment rule $\boldsymbol{p} = \mathbf{p_{i \in N}}$, where $\mathbf{\pi}: T \rightarrow [0, 1]^n$ and $\boldsymbol{p}: T \rightarrow (\mathbb{R}^n)^n$.
		% \end{definition}
	
	%Given all buyers' reported type $\boldsymbol{t}'$, $\pi _{i}(\boldsymbol{t}')$ is the probability that buyer $i$ wins the item and $p_i^j(\boldsymbol{t}')$ is her payment under the circumstances that buyer $j$ wins the item.
	
	% There are some basic properties that a mechanism needs to satisfy, such as a buyer cannot win the item when she is not invited. We define the mechanisms that can work normally as feasible mechanisms.
	
	For the sake of simplicity, we use $\pi_i(\boldsymbol{t})$ or $\pi_i$ to denote the $i$-th element of $\boldsymbol{\pi}(\boldsymbol{t})$, and $p_i^j(\boldsymbol{t})$ or $p_i^j$ to denote the $(i,j)$-th element of $\boldsymbol{p}(\boldsymbol{t})$.
	
	% As a matter of fact, all the diffusion auction mechanisms we mentioned are feasible, since this is the basic requirement of a diffusion auction mechanism.
	
	% Based on the notations above, the expected utility of buyer $i$ subject to mechanism $M=(\boldsymbol{\pi},\boldsymbol{p})$ and $\boldsymbol{t}' \in \boldsymbol{T}$ is defined as follows, where we assume that all buyers are risk neutral and have quasi-linear utility metrics:
	Throughout the paper, we assume that all buyers are risk neutral and have quasi-linear utility metrics, that is, given a mechanism $M=(\boldsymbol{\pi},\boldsymbol{p})$ and a reported type vector $\boldsymbol{t}' \in \boldsymbol{T}$, the expected utility of buyer $i$ is defined as follows.
	\begin{align*}
		\mathbf{E}[u_i(t'_i, \boldsymbol{t}'_{-i}, M)] = \pi _{i} \cdot v_i - \sum_{j=1}^n p_i^j \cdot \pi_j.
	\end{align*}
	We use $u_i$ to denote $u_i(t'_i, \boldsymbol{t}'_{-i}, M)$ for short. In a similar vein, the expected utility (revenue) of the seller is
	\begin{align*}
		\mathbf{E}[u_s(\boldsymbol{t}, M)] = \sum_{i=1}^n \pi_i \cdot \sum_{j=1}^{n} p_j^i = \sum_{i=1}^n \sum_{j=1}^{n} \pi_i \cdot p_j^i.
	\end{align*}
	
	Further, given a mechanism $M$ and a reported type $\boldsymbol{t}' \in \boldsymbol{T}$, the {\em social welfare} is the winner's valuation on the item, denoted by $W(\boldsymbol{t}', M)$, i.e.,$W(\boldsymbol{t}', M)= v_w,$
	where $w\in N$ is the buyer who wins the item.
	
	\paragraph{Properties of mechanisms.} In diffusion auctions, mechanisms are evaluated through a series of desirable properties. Li et al. ~\cite{idm} has provided the basic properties of deterministic mechanisms, which can be easily extended to randomized mechanisms. 
	
	\begin{itemize}
		\item {\bf Feasibility.} A mechanism $M=(\boldsymbol{\pi},\boldsymbol{p})$ is {\em feasible} if the following statements are true for any $\boldsymbol{t}' \in \boldsymbol{T}$,
		\begin{enumerate}[label=\arabic*)] 
			\item If $\boldsymbol{t}'_i = \mathbf{nil}$, then $\pi _{i}(\boldsymbol{t}') = 0$ and $\boldsymbol{p_i}(\boldsymbol{t}') = \mathbf{0}$;
			\item $\sum_{i=1}\nolimits^n \pi _{i}(\boldsymbol{t}') \le 1$. 
		\end{enumerate}
		\item {\bf Individual rationality.}	A diffusion auction mechanism $M$ is {\em individually rational (IR for short)} if for any $i\in N$, $\boldsymbol{t}'_{-i} \in \boldsymbol{T}_{-i}$, and $r'_i \subseteq r_i$, $	\mathbf{E}[u_i((v_i, r'_i), \boldsymbol{t}'_{-i}, M)] \ge 0.$
		\item {\bf Incentive compatibility.} A diffusion auction mechanism $M$ is {\em incentive compatible (IC for short)} if for any $i\in N$, $\boldsymbol{t}'_{-i} \in \boldsymbol{T}_{-i}$, and $t'_i \in \boldsymbol{T}_i$, $ \mathbf{E}[u_i(t_i, \boldsymbol{t}'_{-i}, M)] \ge \mathbf{E}[u_i(t'_i, \boldsymbol{t}'_{-i}, M)]. $
		\item {\bf Weak budget balance.} A diffusion auction mechanism $M$ is {\em weakly budget balanced (WBB for short)} if for any $\boldsymbol{t}' \in \boldsymbol{T}, \mathbf{E}[u_s(\boldsymbol{t}', M)] \ge 0.$
		\item {\bf Efficiency.}	A diffusion auction mechanism $M$ is efficient if for any $\boldsymbol{t}' \in \boldsymbol{T}, \sum_{w \in W}\pi_w(\boldsymbol{t}', M) = 1 $, where $M=(\boldsymbol{\pi},\boldsymbol{p}), W= \{ w \mid w = {\mathop{\arg\max} \nolimits_{i \in N, t_i \ne \nil} v'_i} \}$.
	\end{itemize}
	
	% \begin{itemize}
		%     \item {\bf Feasibility.} A mechanism $M=(\boldsymbol{\pi},\boldsymbol{p})$ is {\em feasible} if the following statements are true for any $\boldsymbol{t}' \in \boldsymbol{T}$,
		% 	\begin{enumerate}[label=\arabic*)] 
			% 		\item If $\boldsymbol{t}'_i = \mathbf{nil}$, then $\pi _{i}(\boldsymbol{t}') = 0$ and $\boldsymbol{p_i}(\boldsymbol{t}') = \mathbf{0}$;
			% 		\item $\sum_{i=1}\nolimits^n \pi _{i}(\boldsymbol{t}') \le 1$. 
			% 	\end{enumerate}
		
		%     \item {\bf Individual rationality.}	A diffusion auction mechanism $M$ is {\em individually rational (IR for short)} if for any $i\in N$, $\boldsymbol{t}'_{-i} \in \boldsymbol{T}_{-i}$, and $r'_i \subseteq r_i$, $	\mathbf{E}[u_i((v_i, r'_i), \boldsymbol{t}'_{-i}, M)] \ge 0.$
		
		%     \item {\bf Incentive compatibility.} A diffusion auction mechanism $M$ is {\em incentive compatible (IC for short)} if for any $i\in N$, $\boldsymbol{t}'_{-i} \in \boldsymbol{T}_{-i}$, and $t'_i \in \boldsymbol{T}_i$, $ \mathbf{E}[u_i(t_i, \boldsymbol{t}'_{-i}, M)] \ge \mathbf{E}[u_i(t'_i, \boldsymbol{t}'_{-i}, M)]. $
		
		%     \item {\bf Weak budget balance.} A diffusion auction mechanism $M$ is {\em weakly budget balanced (WBB for short)} if for any $\boldsymbol{t}' \in \boldsymbol{T}, \mathbf{E}[u_s(\boldsymbol{t}', M)] \ge 0.$
		
		%     \item {\bf Efficiency.}	A diffusion auction mechanism $M=(\boldsymbol{\pi},\boldsymbol{p})$ is efficient if for any $\boldsymbol{t}' \in \boldsymbol{T}, \pi_w(\boldsymbol{t}', M) = 1$, where $w={\mathop{\arg\max} \nolimits_{i \in N, t_i \ne \nil} v'_i}$.
		% \end{itemize}
	
	Apart from the basic properties, Sybil-proofness is a distinctive property in diffusion auctions, requiring the mechanism to be immune to Sybil attacks \cite{sybil}. Formally, a Sybil attack occurs when a buyer $i_0$ creates multiple Sybil identities $S=\{i_1, i_2, \ldots, i_k\}$, each reporting distinct types $t'_{i_1}, t'_{i_2}, \ldots, t'_{i_k}$ respectively \cite{sybil2}. To execute an effective Sybil attack, these identities must satisfy
	\begin{enumerate}[label*=\arabic*)]
		\item For every $i_j\in S$, $r'_{i_j} \subseteq r_{i_0} \cup S$ and $v_{i_j} = v_{i_0}$.
		\item For every $i_j\in S\backslash\{i_0\}$ and $i_l\in N\backslash S$, $i_j\notin r_{i_l}$. 
	\end{enumerate}
	The utility of the Sybil attacker $i_0$ can be defined as $\sum_{j=0}^k \mathbf{E}[u_{i_j}(t'_{i_0}, t'_{i_1},\ldots, t'_{i_k}, \boldsymbol{t}'_{-i_0}, M)]$ using the notation above, as all Sybil identities are under the control of $i_0$. Further, Sybil-proofness requires that no buyer can benefit from Sybil attacks. Formally, a diffusion auction mechanism $M$ is {\bf Sybil-proof} (SP for short) if for any Sybil $i=i_0$ with Sybil identities $i_1, i_2, \ldots, i_k$ and $\boldsymbol{t}'_{-i} \in \boldsymbol{T}_{-i}$,
	\begin{align*}
		\mathbf{E}[u_i(t_i, \boldsymbol{t}'_{-i}, M)] \ge \sum_{j=0}^k \mathbf{E}[u_{i_j}(t'_{i_0}, t'_{i_1},\ldots, t'_{i_k}, \boldsymbol{t}'_{-i}, M)].
	\end{align*}
	
	% \begin{definition}[\bf Sybil-proofness]
		% 	A diffusion auction mechanism $M$ is {\em Sybil-proof (SP for short)} if for any Sybil $i=i_0$ with Sybil identities $i_1, i_2, \ldots, i_k$ and $\boldsymbol{t}'_{-i} \in \boldsymbol{T}_{-i}$,
		% 	\begin{align*}
			% 		\mathbf{E}[u_i(t_i, \boldsymbol{t}'_{-i}, M)] \ge \sum_{j=0}^k \mathbf{E}[u_{i_j}(t'_{i_0}, t'_{i_1},\ldots, t'_{i_k}, \boldsymbol{t}'_{-i}, M)].
			% 	\end{align*}
		% \end{definition}
	
	Notice that SP requires buyers to be truthful, which indicates that any SP mechanism should be IC ~\cite{sybil2}.
	
	% comment begin-------------
	\iffalse
	
	Contrary to Sybil attacks, collusion refers to buyers collude with each other to gain more benefits. In other models, the potential collusive buyers sets (called cartels) are a common knowledge ~\cite{collusion1, collusion2} or only multiple buyers become one buyer is considered ~\cite{collusion3}. We relax the restrictions on cartels, only assuming that their valuations of the item are the same and the corresponding subgraph of the cartel is connected.
	
	\begin{definition}
		A diffusion auction mechanism $M$ is collusion-proof (CP) if for any cartel $C = \{ i_1, i_2, \ldots, i_k \} \subseteq N$ and $t'_{i_j} \in T_{i_j}$, 
		
		\begin{align*}
			\sum_{j=0}^k \mathbf{E}[u_{i_j}(t_{i_0}, t_{i_1},..., t_{i_k}, \boldsymbol{t}'_{-C}, M)] \ge
			\\
			\sum_{j=0}^k \mathbf{E}[u_{i_j}(t'_{i_0}, t'_{i_1},..., t'_{i_k}, \boldsymbol{t}'_{-C}, M)], 
		\end{align*}
		where $t_{i_j} = (v_{i_j}, r_{i_j}), v_{i_1} = v_{i_2} = ... = v_{i_k}$, and the subgraph $G[C]$ is connected.
	\end{definition}
	
	Note that the cartel can be arbitrary and they can misreport or even do not participate to manipulate the auction and gain more benefits, collusion-proof is hard to achieve.
	\fi
	%comment end-------------------
	
	% --------------------------- chap 4 --------------------------
	\section{Mechanism Design and Property Analysis} \label{sec4}
	In this section, we first consider a social network topology---specifically, the path graph---where existing mechanisms achieve very low social welfare. We design a diffusion auction mechanism tailored to this structure and then extend it to general graphs. We subsequently analyze the mechanism's properties, which include IR, IC, WBB, SP, collusion-proofness, approximate efficiency, and the seller's revenue.
	
	\subsection{Probabilistic Diffusion Mechanism} \label{sec3}
	When the network is a path graph, existing mechanisms frequently allocate the item to the first buyer, potentially resulting in severely suboptimal social welfare. This subsection therefore prioritizes the path graph topology, proposing a novel mechanism named probabilistic diffusion mechanism (PDM for short). Through theoretical analysis, we show that PDM is IR, IC, WBB, and SP.
	
	Without loss of generality, suppose that the path graph is from $1$ to $n$, i.e., the edge set is $E = \{(1, 2), (2, 3), ..., (n-1, n) \}$, and the seller $s$ only knows the buyer $1$. Under these settings, the PDM mechanism is defined as follows.
	
	\begin{definition} [\bf Probabilistic Diffusion Mechanism]
		
		Let $v^*_B = \max \nolimits_{i \in B} v'_i $ and $N_{-i}$ be the set of buyers who can participate in the auction when $i$ is not invited, i.e., $N_{-i} = \{1, 2, \ldots, i-1 \}$. Given $\boldsymbol{t}' \in T$, the allocation rule of probabilistic diffusion mechanism for buyer $i$ is defined as:
		\begin{align*}
			\pi_i = 
			\begin{cases} 
				\max \{0, v'_i - v^*_{N_{-i}} \},  & \text{if} \; i>1, \\
				1-v^*_N+v'_1,  & \text{if} \; i=1. \\
			\end{cases}
		\end{align*}	
		If buyer $1$ wins the item, no one needs to pay. If buyer $j>1$ wins the item, the payment rule for buyer $i$ is defined as:
		\begin{align*}
			p_i^j = 
			\begin{cases} 
				\frac{v^*_{N_{-j}} + v'_j}{2},  & \text{if} \; i=j, \\
				-\frac{v^*_{N_{-j}} + v'_j}{2},  & \text{if} \; i=1, \\
				0, & \text{otherwise.}
			\end{cases}
		\end{align*}	
	\end{definition}
	
	Intuitively, if a buyer's bid is not higher than the highest bid before her, she does not have a chance to win the item. Since buyer $1$'s reward equals to the winner's payment, the seller's revenue is $0$, which is the same as other weakly budget balanced mechanisms. However, the expected social welfare is not less than other weakly budget balanced mechanisms, since PDM may assign the item to the buyer whose bid is higher than $v'_1$.
	
	We now present a running example of PDM. The social network structure is shown in Figure \ref{eg_PDM}, with four buyers having values of $0.2$, $0.1$, $0.4$, and $1$, respectively. Under the PDM, note that buyer $b$'s bid is lower than that of buyer $a$, and thus buyer $b$ will not win the item, while buyers $a$, $c$, and $d$ win the item with probabilities of $0.2$, $0.2$, and $0.6$, respectively. When buyer a wins the item, no payment is required; when buyers $c$ or $d$ win, buyer a receives referral rewards of $0.3$ or $0.7$, respectively, for inviting others. Therefore, the expected social welfare is $\mathbf{E}[W(\boldsymbol{t}, PDM)] = 0.2 \times 0.2 + 0.2 \times 0.4 + 0.6 \times 1 = 0.72$. In comparison, other mechanisms such as IDM and CDM achieve a social welfare equal to buyer $a$'s value, i.e., $0.2$.
	%接下来我们给出一个PDM的运行实例。社会网络结构如图1所示，一共有四个买家，价值分别为0.2、0.1、0.4和1。在PDM下，注意到买家b的出价低于a的出价，因此买家b不会获得物品，买家a、c、d分别有0.2、0.2和0.6的概率赢得物品。买家a在自己获得物品时，无需支付，在买家c、d获得物品时，分别会得到0.3和0.7的邀请其他人的奖励。因此，期望社会福利$E(W)=0.2 \times 0.2 + 0.2 \times 0.4 + 0.6 \times 1 = 0.72$，相较之下，其他机制如IDM、CDM，他们的社会福利为买家a的价值，即0.2。
	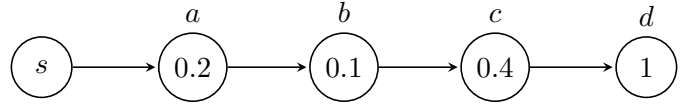
\begin{figure}[t!]
		\begin{center}
			\begin{tikzpicture}[->,>=stealth,shorten >=1pt,auto, semithick]
				
				\node at (0,0) [circle, draw, minimum width=0.8cm, minimum height=0.8cm] (s) {$s$};
				\node at (2, 0) [circle, draw, label=above:$a$, minimum width=0.8cm, minimum height=0.8cm] (a)  {$0.2$};
				\node at (4, 0) [circle, draw, label=above:$b$, minimum width=0.8cm, minimum height=0.8cm] (b) {$0.1$};
				\node at (6, 0) [circle, draw, label=above:$c$, minimum width=0.8cm, minimum height=0.8cm] (c) {$0.4$};
				\node at (8, 0) [circle, draw, label=above:$d$, minimum width=0.8cm, minimum height=0.8cm] (d) {$1$};
				
				\path (s) edge node {} (a)
				(a) edge node {} (b)
				(c) edge node {} (d)
				(b) edge node {} (c);
			\end{tikzpicture}
			\caption{Example of a social network. There are four buyers in the network and buyer $d$ has the highest valuation $1$.}
			\label{eg_PDM}
		\end{center}
	\end{figure}
	
	From its definition, it is not hard to see that PDM is feasible, IR, and WBB. Moreover, the following theorem states that PDM is SP.
	%定理的意义
	\begin{theorem} \label{Thm1}
		PDM is SP.
	\end{theorem}
	
	\begin{proof}
		In this case, we suppose the network after buyer $i$ carrying out a Sybil attack is still a path graph since PDM performs on path graphs and we will prove the general case in Section \ref{sec4}.
		
		First, we prove the buyer $i \; (i>1)$ whose bid is not higher than the highest bid before her will truthfully report. If she truthfully reports, she will not win the item and her utility $u_i$ is zero. Note that $i>1$, so she cannot get any reward. The only way to have positive utility is to bid higher than $v^*_{N_{-i}}$ and luckily win the item. However, if she wins, she will pay at least $\frac{v^*_{N_{-i}} + v'_i}{2} \ge v^*_{N_{-i}} \ge v_i$, so her utility $u'_i \le 0$. Therefore, she does not have motivation to misreport. Hence, we can ignore such buyers and assume the buyers' bids are increasing.
		
		Second, we prove buyer $1$ will not increase her expected utility by Sybil attack. Suppose buyer $i_0 = 1$ creates some Sybil identities $i_1, i_2, \ldots, i_k$ who report $t'_{i_1}, t'_{i_2}, \ldots, t'_{i_k}$ respectively. If $v'_{i_{j-1}} \ge v'_{i_j}$ for $j \in \{ 1, 2, \ldots, k\} $, then $i_j$ can be ignored, so we can assume $v'_{i_{j-1}} < v'_{i_j}$ for every $j \in \{ 1, 2, \ldots, k\} $. Let $S = \{ i_0, i_1, \ldots, i_k \}$. If $v'_{i_k}$ is not the highest bid, the total expected utility of buyer $1$ and her Sybil identities is
		\begin{align*}
			\mathbf{E}[u_S]
			=&\sum_{j=0}^k \mathbf{E}[u_{i_j}(t'_{i_0}, t'_{i_1},\ldots, t'_{i_k}, \boldsymbol{t}'_{-1}, M)] \\
			= & \; \mathbf{E}[u_{S} \mid w \in S] + \mathbf{E}[u_{S} \mid w \notin S]\\
			= &\; ((1-v'_m+v'_{i_0}) + \sum_{j=1}^k (v'_{i_j}-v'_{i_{j-1}})) \cdot v_1 \; + 
			\; \sum_{j=2}^m ((v'_j-v'_{{j-1}}) \cdot \frac{v'_j+v'_{{j-1}}}{2})\\
			%= &\; (1-v'_m+v'_{i_0}) \cdot v_1 + \frac{v'^2_m - v'^2_{i_0}}{2} + \\
			%&\; (v'_{i_k}-v'_{i_{0}}) \cdot v_1 -  \frac{v'^2_k - v'^2_{i_0}}{2} \\
			= &\; (1-v'_m+v'_{i_k}) \cdot v_1 + \frac{v'^2_m - v'^2_{i_k}}{2}, 
		\end{align*}
		where $v'_1 = v'_{i_k}$, $w$ is the winner, and $v'_m$ is the highest bid. When the highest bid from the others is higher than $v_1$, we take the derivative of $\mathbf{E}[u_S]$ with respect to $v'_{i_k}$, $\mathbf{E}'[u_S] = v_1 - v'_{i_k}$. Therefore, her expected utility is maximized when she bids $v'_{i_k} = v_1$.
		$ \mathbf{E}[u_S]$ has maximum value $v_1 + \frac{(v'_m - v_{1})^2}{2}$, which equals to $\mathbf{E}[u_1(t_1, \boldsymbol{t}'_{-1}, M)]$. If she does not invite buyer $2$ or $v_{i_k}$ is the highest bid, her utility $u'_1 = v_1 \le v_1 + \frac{(v'_m - v_{1})^2}{2}$. When the highest bid from the others is not higher than $v_1$, she still maximizes her expected utility when truthfully reporting because her expected utility can never higher than $v_1$. Therefore, buyer $1$ will not increase her expected utility by using Sybil identities.
		
		Finally, we prove buyer $i > 1$ will not increase her expected utility by Sybil attack. Suppose buyer $i = i_0$ creates some Sybil identities $i_1, i_2, \ldots, i_k$ who report $t'_{i_1}, t'_{i_2}, \ldots, t'_{i_k}$ respectively. Similarly, we suppose $v'_{i_{j-1}} < v'_{i_j}$ for every $j \in \{ 1, 2, \ldots, k\} $. The maximum expected utility of buyer $i$ and her Sybil identities can be deduced by differentiation in the same way.  
		\begin{align*}
			\sum_{j=0}^k \mathbf{E}[u_{i_j}] 
			%	= &\; \sum_{j=0}^k (v'_{i_j}-v'_{i_{j-1}}) \cdot v_i -  \frac{v'^2_{i_j}-v'^2_{i_{j-1}}}{2} \\
			= &\; (v'_{i_k}-v'_{i-1}) \cdot v_i - \frac{v'^2_{i_k} - v'^2_{i-1}}{2} \\
			\le &\; \frac{(v_{i} - v'_{i-1})^2}{2} \\
			= &\; \mathbf{E}[u_i(t_i, \boldsymbol{t}'_{-i}, M)].
		\end{align*}
		%where $v'_{i_{-1}} = v'_{i-1}$. 
		Therefore, buyer $i>2$ will not increase their expected utility by using Sybil identities and thus PDM is SP. \qed
	\end{proof} 
	
	Since SP implies IC, we have the following corollary.
	
	\begin{corollary}
		PDM is IC.
	\end{corollary}
	
	This result implies that when the social network is a path graph, PDM satisfies the desired properties, which lays a foundation for subsequent proofs of IC and SP for mechanisms extended to general network structures.
	%这个结果意味着当社会网络是一个路径图时，PDM能够满足所期望的性质，为接下来扩展到一般的社会网络结构的机制IC和SP的证明提供了很大帮助。
	
	\subsection{\boldmath $f$-PDM on Graphs}
	
	The basic idea to generalize our mechanism is mapping the graph to a path graph. However, due to randomness, the range of the map should be a probability distribution of path graphs. Formally, we define $\Delta$ as the set of probability distribution over permutations of $N$. Without loss of generality, we assume $G=(N,E) \in \mathcal{G}$ is connected, i.e., every buyer can participate in the auction, where $\mathcal{G}$ is the set of network $G$.
	
	Different maps result in corresponding mechanisms having different properties. Therefore, we propose a class of diffusion mechanisms according to different maps. We define the properties of a map first.
	
	In ~\cite{idm}, for any $i, j \in N$, we say that $i$ is the \emph{diffusion critical node} of $j$ when $j$ cannot participate in the auction if $i$ is not invited. Formally,
	all the paths from the seller $s$ to $j$ have to pass $i$, denoted by $i \preccurlyeq j$. To ensure incentive compatibility, this order must be remained.
	
	Different from the definition in ~\cite{idm}, for any $i \in N$, $i$ is the diffusion critical node of its own, i.e., $i \preccurlyeq i$. Intuitively, if $i$ is $j$'s diffusion critical node in the original network, $i$ should still be $j$'s diffusion critical node after the map. We define such maps as order-preserving maps. Formally, a map $f: \mathcal{G} \rightarrow \Delta$ is an order-preserving map if for any $G=(N,E) \in \mathcal{G}$, possible outcome $o = (i_1, i_2, \ldots, i_n)$ and $j, k \in N$, if $j \preccurlyeq k$ then $i_j \le i_k$.
	
	As a matter of fact, the outcome of an order-preserving map is a topological ordering with respect to the partial order $\preccurlyeq$. However, this is not enough to ensure IC, unless every buyer achieves the best outcome when diffusing honestly. To make comparisons between the outcomes (i.e., distributions over permutations of $N$), we utilize the notion \emph{stochastic dominance} ~\cite{stochasticdominance,stochasticdominance2}. Formally, given buyer $i$ and two probability distributions over permutations of $N$, denoted as $\mu_1, \mu_2$, let $q^l_{A}= \Pr \nolimits_{\sigma \sim \mu_l} [ \{ \sigma(j) \mid j \in A \} = \{ 1, 2, \ldots, \sigma(i)-1 \}]$, where $l = 1, 2$ and $A \subseteq N \backslash \{i\}$. Then $\mu_1$ stochastically dominates $\mu_2$ for $i$ if and only if
	\begin{align*}
		\sum_{B \subseteq A} q^1_{B} \ge \sum_{B \subseteq A} q^2_{B},\quad \text{for all } A \subseteq N \backslash \{i\}.
	\end{align*}

	Intuitively, if a distribution $p$ stochastically dominates $q$ for buyer $i$, the utility of $i$ under $p$ is higher than that under $q$. Therefore, to guarantee incentive compatibility, the map should achieve the property that if for every buyer $i$, the distribution stochastically dominates other distributions when she diffuses honestly. We formalize such maps as follows.
	
	\begin{definition}
		A map $f: \mathcal{G} \rightarrow \Delta$ is an incentive diffusion map if $f$ is an order-preserving map, and for any $G=(N,E) \in \mathcal{G}$, $i \in N$ and any subgraph $G'$ of $G$, $f(G)$ stochastically dominates $f(G')$ for $i$, where $G'$ is the remaining network whose vertices can participate in the auction when some of $i$'s outgoing edges are removed. 
	\end{definition}
	
	The property of incentive diffusion maps ensures that buyers will not achieve better outcome when they do not report their neighbors honestly. Based on the incentive diffusion maps, we show our generalized mechanism called $f$-PDM.
	
	\begin{definition} [$f$-PDM]
		Given an incentive diffusion map $f: \mathcal{G} \rightarrow \Delta$, $f$-probabilistic diffusion mechanism ($f$-PDM for short) works as follows.
		
		\begin{enumerate}[label*=\arabic*)]
			\item Map the graph $G$ to the path graph drawn from $f(G)$.
			\item Perform PDM on the path graph.
			\item Let $i$ be the first buyer in the path graph. She pays extra $\frac{1}{2}(v^*_{N_{-i}})^2$, where $v^*_B = \max \nolimits_{j \in B} v'_j $ and $N_{-i} = N \, \backslash \, \{ \, j \, | \, i \preccurlyeq j \, \}$ is the set of buyers who can participate in the auction when $i$ is not invited.
		\end{enumerate}
		
	\end{definition}

	\begin{figure}[t!]
		\begin{center}
			\begin{tikzpicture}[->,>=stealth,shorten >=1pt,auto, semithick]
				
				\node at (0,0) [circle, draw, minimum width=0.8cm, minimum height=0.8cm] (s) {$s$};
				\node at (2, 0.8) [circle, draw, label=above:$a$, minimum width=0.8cm, minimum height=0.8cm] (a)  {$0.3$};
				\node at (2, -0.8) [circle, draw, label=below:$b$, minimum width=0.8cm, minimum height=0.8cm] (b) {$0$};
				\node at (4, 0.8) [circle, draw, label=above:$c$, minimum width=0.8cm, minimum height=0.8cm] (c) {$0.9$};
				
				\path (s) edge node {} (a)
				(s) edge node {} (b)
				(a) edge node {} (b)
				(a) edge node {} (c)
				(b) edge node {} (a);
			\end{tikzpicture}
			\caption{Example of a social network. There are three buyers in the network and buyer $c$ has the highest valuation $0.9$.}
			\label{fig1}
		\end{center}
	\end{figure}
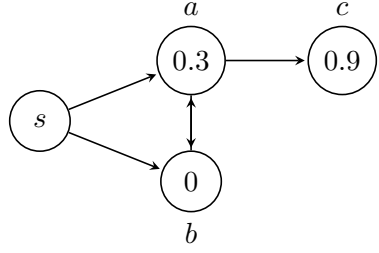
	
	\begin{figure}[t!]
		\begin{center}
			\begin{tikzpicture}[->,>=stealth,shorten >=1pt,auto, semithick]
				\node at (-1.5,3) (n2) { Case 1: };
				\node at (0,3) [circle, draw, minimum width=0.8cm, minimum height=0.8cm] (s2) {$s$};
				\node at (1.5, 3) [circle, draw, label=above:$a$, minimum width=0.8cm, minimum height=0.8cm] (a2)  {$0.3$};
				\node at (3, 3) [circle, draw, label=above:$b$, minimum width=0.8cm, minimum height=0.8cm] (b2) {$0$};
				\node at (4.5, 3) [circle, draw, label=above:$c$, minimum width=0.8cm, minimum height=0.8cm] (c2)  {$0.9$};
				\node at (-1.5,1.5) (n1) { Case 2: };
				\node at (0,1.5) [circle, draw, minimum width=0.8cm, minimum height=0.8cm] (s1) {$s$};
				\node at (1.5, 1.5) [circle, draw, label=above:$a$, minimum width=0.8cm, minimum height=0.8cm] (a1)  {$0.3$};
				\node at (3, 1.5) [circle, draw, label=above:$c$, minimum width=0.8cm, minimum height=0.8cm] (c1) {$0.9$};
				\node at (4.5, 1.5) [circle, draw, label=above:$b$, minimum width=0.8cm, minimum height=0.8cm] (b1)  {$0$};
				\node at (-1.5,0) (n2) { Case 3: };
				\node at (0,0) [circle, draw, minimum width=0.8cm, minimum height=0.8cm] (s3) {$s$};
				\node at (1.5, 0) [circle, draw, label=above:$b$, minimum width=0.8cm, minimum height=0.8cm] (b3)  {$0$};
				\node at (3, 0) [circle, draw, label=above:$a$, minimum width=0.8cm, minimum height=0.8cm] (a3) {$0.3$};
				\node at (4.5, 0) [circle, draw, label=above:$c$, minimum width=0.8cm, minimum height=0.8cm] (c3)  {$0.9$};
				
				\path (s1) edge node {} (a1)
				(a1) edge node {} (c1)
				(c1) edge node {} (b1)
				(s2) edge node {} (a2)
				(a2) edge node {} (b2)
				(b2) edge node {} (c2)
				(s3) edge node {} (b3)
				(b3) edge node {} (a3)
				(a3) edge node {} (c3);
			\end{tikzpicture}
			\caption{Example of $f$-PDM. There are 3 cases that $f$ might map to under the network in Figure \ref{fig1}. Then PDM is performed.}
			\label{fig2}
		\end{center}
	\end{figure}
	
	Intuitively, $f$-PDM maps the graph to a path graph first, and then performs PDM on the path graph. Moreover, to avoid zero revenue of the seller, the first buyer is additionally charged. We give a running example of $f$-PDM for a better understanding. Figure \ref{fig1} shows an example of a network. First, $f$ maps the network to one of the path graphs in Figure \ref{fig2}. The constraint that buyer $a$ must precede $c$, imposed by $f$ being an incentive diffusion map, leaves only the three configurations in Figure \ref{fig2} as feasible.
	
	Different maps correspond to different probability distributions of the three cases. We now present two mapping examples---the breadth-first map and the generalized breadth-first map---defined by Mechanisms \ref{bfmap} and \ref{f_g} respectively. The properties of these two maps will be analyzed in Section \ref{SecSP}.
	
	\begin{algorithm}[h]
		\caption{The breadth-first map}
		\label{bfmap}
		\KwIn {A graph $G=(N, E)$, neighbors of the seller $r_s$}
		\KwOut{A permutation of $N$}
		$L_i \gets \{ v \mid \text{the distance from $s$ to $v$ is } i. \}$
		
		\For{$i=1$ to $|N|$} {
			\While{$L_i \neq \emptyset$}{
				Select an element $j$ in $L_i$ with equal probability \; 
				Output $j$ \tcp*{$j$ is the selected buyer.}
				$L_i \gets L_i - \{ j \} $\;
			}
		}
	\end{algorithm}
	
	\begin{algorithm}[h]
		\caption{The generalized breadth-first map}
		\label{f_g}
		\KwIn {A graph $G=(N, E)$}
		\KwOut{A permutation of $N$}
		$L \gets r_s $ \tcp*{$L$ is the candidate set, $r_s$ is the set of neighbors of the seller.}
		$A \gets \emptyset $ \tcp*{$A$ is the set of selected buyers.}
		\While{$L \neq \emptyset$} {
			Select an element $j$ in $L$ with equal probability \; 
			Output $j$ \tcp*{$j$ is the selected buyer.}
			$A \gets A \cup \{ j \} $ \;
			$L \gets (L - \{ j \}) \cup (r_j - A) $\;
		}
	\end{algorithm}
	
	%接下来我们给出两个映射实例—— breadth-first map和generalized breadth-first map，分别由机制1和机制2定义。
	
	For the breadth-first map, it selects buyers sequentially based on their distance from the seller. Consequently, buyers $a$ and $b$ each have a $0.5$ probability of being selected first, with the remaining buyer becoming second, corresponding to case $1$ and case $3$ in Figure \ref{fig2} respectively.
	%它按照距离卖家的顺序选择买家，因此，买家a和买家b各有0.5的概率成为第一个买家，另外一个成为第二个买家，分别对应图1的case 1和case 3。
	
	In contrast, the generalized breadth-first map uniformly selects the current buyer $j$ at random from the candidate set, after which it adds $j$'s neighbors to the candidate set. Similarly, buyers $a$ and $b$ each have a probability of $0.5$ of being selected first. However, once buyer $a$ becomes the first buyer, buyer $c$ is added to the candidate set and then competes with buyer $b$ with equal probability for the second selection. Therefore, the probabilities of Cases $1$, $2$, and $3$ occurring in Figure \ref{fig2} are $0.25$, $0.25$, and $0.5$ respectively.
	%相较而言，generalized breadth-first map每次从候选人的集合中等概率地选取当前买家j，之后将j的邻居加入到候选人集合中。同样，买家a和买家b各有0.5的概率成为第一个买家，然而一旦买家a首先被选中，买家c加入到候选人集合中，（买家c）与b等概率地成为第二个买家。因此图二中的case1、2、3发生的概率分别为0.25、0.25、0.5。
	
	In case $1$, $a$ becomes the first buyer. Since $v_b$ is not higher than $v_a$, $b$ has no chance to win, while the winning probabilities for $a$ and $c$ are $0.4$ and $0.6$ respectively. The expected social welfare is $(0.3+0.1) \times 0.3 + (0.9-0.3) \times 0.9 = 0.66$ and the seller's revenue is $0$. In case $2$, the expected social welfare is $0.66$ and the seller's revenue is $0$. In case $3$, the expected social welfare is $0.63$ and the seller's revenue is $0.405$. 
	By comparison, the expected social welfare of IDM is $0.3$ and the seller's revenue of IDM is $0$.
	
	At the end of this subsection, we analyze the time complexity of f-PDM. Since $f$-PDM represents a class of mechanisms, its specific time complexity depends on the map $f$. In particular, the mechanisms corresponding to the breadth-first map and the generalized breadth-first map have a time complexity equivalent to that of breadth-first graph traversal, which is $O(|E \cup r_s|)$. (This holds because we assume the graph is connected, so $|E| \cup r_s \ge n$.)
	%在本小节的最后，我们阐述f-PDM的时间复杂度。由于f-PDM是一类机制，因此具体的时间复杂度取决于映射f。特别地，The breadth-first map和The generalized breadth-first map对应的机制的时间复杂度相当于广度优先遍历图的复杂度，都是O(|E|) (因为我们假设图是连通的，因此|E|>=n-1，并且r_s<n)。

	\subsection{Individual Rationality and Incentive Compatibility}
	%\subsection{IR, IC and Sybil-proofness}
	
	It is not hard to see that $f$-PDM is WBB, we prove $f$-PDM is IR first. %由于IR是IC的基本推论，因此我们先将$f$-PDM是IR的作为引理证明，来为IC的证明打好基础。
	
	\begin{lemma} \label{IR}
		For any incentive diffusion map $f$, $f$-PDM is IR.
	\end{lemma}
	
	\begin{proof}
		Note that PDM is individually rational from Section \ref{sec3}, we only need to prove the first buyer after the map has non-negative utility. Suppose $i$ is the first buyer after the map, then her expected utility under PDM is $v_i+\frac{1}{2}(v^*_N-v_i)^2$. Therefore, her expected utility under $f$-PDM is
		\begin{align*}
			\mathbf{E}[u_i] 
			= &\; v_i+\frac{1}{2}(v^*_N-v_i)^2-\frac{1}{2}(v^*_{N_{-i}})^2
			\ge \; v_i+\frac{1}{2}(v^2_i-2 v_i  v^*_N) \\
			= &\; v_i (1-v^*_N) +\frac{1}{2} v^2_i \;
			\ge \; \frac{1}{2} v^2_i \;
			\ge \; 0. 
		\end{align*}
		Therefore, $f$-PDM is IR. \qed
	\end{proof}
	
	Furthermore, we prove $f$-PDM is IC.
	
	\begin{theorem}
		For any incentive diffusion map $f$, $f$-PDM is IC.
	\end{theorem}
	
	\begin{proof}
		Since PDM is IC and $f(G)$ is independent to the bids, we only need to show every buyer will report her neighbors truthfully. Note that $f$ is an incentive diffusion map, so for every buyer $i$, $f(G)$ stochastically dominates $f(G')$, where $G'$ is the network when $i$ misreports. Let $u_i$ and $u'_i$ be the expected utility under $G$ and $G'$ respectively. We prove $u_i \ge u'_i$ for any $\boldsymbol{t}'_{-i} \in \boldsymbol{T}_{-i}$.
		
		Suppose $i_1, i_2, \ldots, i_k$ are buyers whose bids are lower than $v'_i$ and $v'_{i_1} \le v'_{i_2} \le \cdots \le v'_{i_k} \le v'_i$. Let $q^l_{A}= \Pr \nolimits_{\sigma \sim \mu_l} [ \{ \sigma(j) \mid j \in A \} = \{ 1, 2, \ldots, \sigma(i)-1 \}]$, where $l = 1, 2$, $\mu_1 = f(G)$, $\mu_2 = f(G')$, and $A \subseteq N \backslash \{i\}$. Moreover, let $\mathcal{A}_j = \{ A \subseteq N \backslash \{i\} \mid i_j $ is the largest buyer among the highest bidders in $ A \} $ ($\mathcal{A}_0 = \{ \emptyset \}$),  and $q^l_j = \sum_{A \in \mathcal{A}_j} q^l_{A}$. Let $u^j_{i}$ be the expected utility of $i$ under PDM when $\max \nolimits_{\sigma(m) < \sigma(i)} v_m =   v_{i_j}$ ($u_i^0$ is the expected utility of $i$ under PDM when $i$ is the first buyer after the map), then from the proof of Lemma \ref{IR}, $u^0_i \ge \frac{1}{2} v^2_i \ge \frac{1}{2} ( v_i - v'_{i_1})^2 = u^1_i$. Therefore, $u^0_{i} \ge u^1_i \ge \cdots \ge u^k_i \ge 0$.
		
		Since $\mu_1 = f(G)$ stochastically dominates $\mu_2 = f(G')$, for every $B_j = \{ i_1, i_2, \ldots, i_j \} \subseteq N \backslash \{i\} (B_0 = \emptyset)$, $\sum_{l=0}^{j}q^1_l = \sum_{B \subseteq B_j} q^1_{B} \ge \sum_{B \subseteq B_j} q^2_{B}=\sum_{l=0}^{j}q^2_l$. Therefore, $u_i = \sum_{l=0}^k q^1_l \cdot u^l_i \ge \sum_{l=0}^k q^2_l \cdot u^l_i = u'_i$ by iteration, and thus her expected utility is maximized when she reports her neighbors truthfully.
		
		Therefore, $f$-PDM is IC. \qed
	\end{proof}
	
	Moreover, $f$-PDM can also satisfy ex-post properties. Ex-post properties ensure that after the result is revealed, everyone will participate in the auction and report the same type if they start over. $f$-PDM is ex-post WBB, because the seller's revenue is always non-negative. For ex-post IC, even if all types are revealed, every buyer will maximize her utility by reporting truthfully when starting over, because she is risk neutral.
	
	For ex-post IR, it can be satisfied with a minor adjustment, and it still ensures other properties such as IC and approximate efficiency. Note that only the first buyer may violate ex-post IR because she has to pay extra $(v^*_{N_{-i}})^2/2$. We can allocate her payment to each situation and she will pay $(v^*_{N_{-i}})^2/2$ in expectation. For example, letting $k-1 < N_{-i} \le k$, we can change the payment rule to $p^1_1 = v^2_1/(2 \pi_1), p_1^i=0 \, (2 \le i \le k-1), p_1^k=\frac{(v^*_{N_{-i}})^2-v_k^2}{2\pi_k}$, the first buyer does not need to pay extra, and the other rules remain unchanged.
	
	\subsection{Sybil-proofness} \label{SecSP}
	
	%女巫攻击对于扩散拍卖有着很大的威胁，因为买家可以通过创造多个女巫身份来增加获胜的可能或得到的奖励，进而增加自己的效用，这损害了其他买家和卖家的利益。然而，目前除了少数机制之外，很少有机制能够保证Sybil-proof。例如，对于图3，如果买家a诚实报告，那么在IDM机制下他的效用为0。如果她进行女巫攻击，创造一个女巫身份a'，a'报价0.9，那么她将获得0.8的奖励，而买家c的支付从0.1变成了0.9。
	Sybil attacks pose a significant threat to diffusion auctions, as buyers can increase their winning probability or receive more rewards by creating multiple Sybil identities, thereby increasing their utility, which harm the interests of both other buyers and the seller. However, few existing mechanisms guarantee Sybil-proofness, with only a handful of exceptions. For example, in Figure \ref{figsp}, if buyer $a$ reports truthfully, their utility under the IDM mechanism is $0$. If she launches a Sybil attack by creating a Sybil identity $a'$ with a reported valuation of $0.9$, she would receive a reward of $0.8$, while buyer $c$'s payment increases from 0.1 to 0.9.

	\begin{figure}[t!]
		\begin{center}
			\begin{tikzpicture}[->,>=stealth,shorten >=1pt,auto, semithick]
				
				\node at (-6,0) [circle, draw, minimum width=0.8cm, minimum height=0.8cm] (s) {$s$};
				\node at (-4, 0.8) [circle, draw, label=above:$a$, minimum width=0.8cm, minimum height=0.8cm] (a)  {$0$};
				\node at (-4, -0.8) [circle, draw, label=below:$b$, minimum width=0.8cm, minimum height=0.8cm] (b) {$0.1$};
				\node at (-2, 0.8) [circle, draw, label=above:$c$, minimum width=0.8cm, minimum height=0.8cm] (c) {$1$};
				\node at (2,0) [circle, draw, minimum width=0.8cm, minimum height=0.8cm] (s1) {$s$};
				\node at (4, 0.8) [circle, draw, label=above:$a$, minimum width=0.8cm, minimum height=0.8cm] (a1)  {$0$};
				\node at (4, -0.8) [circle, draw, label=below:$b$, minimum width=0.8cm, minimum height=0.8cm] (b1) {$0.1$};
				\node at (6, 0.8) [circle, draw, label=above:$c$, minimum width=0.8cm, minimum height=0.8cm] (c1) {$1$};
				\node at (6, -0.8) [circle, draw, label=below:$a'$, minimum width=0.8cm, minimum height=0.8cm] (d1) {$0.9$};
				\node at (-4,-2) (n2) { (a) };
				\node at (4,-2) (n2) { (b) };
				
				\path (s) edge node {} (a)
				(s) edge node {} (b)
				(a) edge node {} (c)
				(s1) edge node {} (a1)
				(s1) edge node {} (b1)
				(a1) edge node {} (c1)
				(a1) edge node {} (d1);
			\end{tikzpicture}
			\caption{Example of a Sybil attack. Buyer $a$ is incentivized to create a Sybil identity $a'$, and get a reward of $0.8$ instead of $0$.}
			\label{figsp}
		\end{center}
	\end{figure}
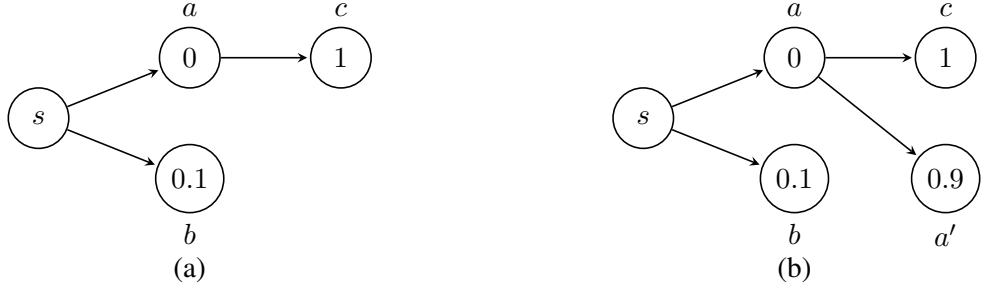

	Similar to IDM, $f$-PDM is not SP in general, since creating Sybil-identities may cause a distribution that stochastically dominates the original one. Therefore, we need to improve the map and define maps that are Sybil-proof.
	
	\begin{definition}
		A map $f: \mathcal{G} \rightarrow \Delta$ is a Sybil-proof map if $f$ is an incentive diffusion map, and for any $G=(N,E) \in \mathcal{G}$, $i \in N$ and any subgraph $G'$ of $G$, $f(G')$ stochastically dominates $f(G)$, where $G'$ is the remaining network when some of the nodes $i_1, \ldots , i_k$ are removed and for every $1 \le j \le k$, $i$ is $i_j$'s diffusion critical node.
	\end{definition}
	
	Intuitively, Sybil-proof maps indicate that Sybil identities of $i$ do not affect her outcome. We formalize the property as follows.
	
	\begin{theorem} \label{SPmap}
		When $f$ is a Sybil-proof map, $f$-PDM is SP.
	\end{theorem}
	
	\begin{proof}
		Given buyer $i$, since $f$ is a Sybil-proof map, creating Sybil identities does not affect the buyer's outcome. Moreover, $f$ is order-preserving, so $i$ is still smaller than her Sybil identities after the map, which makes Sybil identities useless. Therefore, $f$-PDM is SP. \qed
	\end{proof}
	
	At the end of this subsection, we demonstrate that both the breadth-first map and the generalized breadth-first map are Sybil-proof maps, and thus the corresponding mechanisms are SP. Consequently, they are also IC, since SP implies IC.
	
	\begin{proposition}
		The Sybil-proofness of $f$-PDM are as follows.
		
		\begin{enumerate}[label*=\arabic*)]
			\item For the generalized breadth-first map $f_g$, $f_g$-PDM is SP.
			\item When $f$ is the breadth-first map, $f$-PDM is SP.
		\end{enumerate}
		
	\end{proposition}
	
	\begin{proof}
		By Theorem \ref{SPmap}, we only need to prove $f_g$ and the breadth-first map $f$ are Sybil-proof maps.
		
		It is easy to check that $f_g$ and $f$ are order-preserving maps. Moreover, for any buyer $i$ and any set of other buyers $A$, the probability $q_{A}= \Pr \nolimits_{\sigma \sim f_g} [ \{ \sigma(j) \mid j \in A \} = \{ 1, 2, \ldots, \sigma(i)-1 \}]$ remains unchanged when $i$ does not invite some of her neighbors and creates Sybil identities. The reason is as follows. Suppose $A=\{ a_1, a_2, \ldots, a_{k-1}\},  \, p_1 = \Pr \nolimits_{\sigma \sim f_g}[\sigma(a_1)=1], \, p_j = \Pr \nolimits_{\sigma \sim f_g}[\sigma(a_j)=j \mid \sigma(a_1)=1, \ldots, \sigma(a_{j-1})=j-1] \, (2 \le j \le k)$, where $a_k=i$. 
		Then removing the outgoing edges of $i$ and creating Sybil nodes do not change $p_j \, (1 \le j \le k)$ due to the characteristics of $f_g$ and $f$. Therefore, $q_{A}$ remains unchanged and thus $f_g$ and $f$ are Sybil-proof maps.
		
		Therefore, $f_g$-PDM and $f$-PDM are both SP. \qed
	\end{proof}
	
	Moreover, for the generalized breadth-first map, when selecting the $i$-th buyer in $L$, the probability can be positively correlated to the number of their neighbors. In this case, the mechanism is IC but no longer SP, in exchange for more incentives to invite neighbors.
	
	\subsection{Approximate Efficiency}
	
	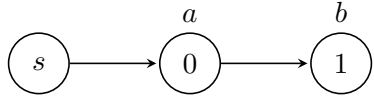
\begin{figure}[t!]
		\begin{center}
			\begin{tikzpicture}[->,>=stealth,shorten >=1pt,auto, semithick]
				
				\node at (0,0) [circle, draw, minimum width=0.8cm, minimum height=0.8cm] (s) {$s$};
				\node at (2, 0) [circle, draw, label=above:$a$, minimum width=0.8cm, minimum height=0.8cm] (a)  {$0$};
				\node at (4, 0) [circle, draw, label=above:$b$, minimum width=0.8cm, minimum height=0.8cm] (b) {$1$};
				
				\path (s) edge node {} (a)
				(a) edge node {} (b);
			\end{tikzpicture}
			\caption{Example of inefficiency. The seller $s$ only knows buyer $a$ with valuation $0$. Buyer $a$ only knows buyer $b$ whose valuation is $1$ and buyer $b$ does not know anyone. In this case, the existing mechanisms assign the item to $a$ for free, so the social welfare is $0$.}
			\label{fig3}
		\end{center}
	\end{figure}
	
	Li et al. ~\cite{ic} has shown that IR, IC, and weakly budget balanced diffusion mechanisms are not efficient. Therefore, a natural idea is to seek for the approximate version of efficiency.
	
	\begin{definition}
		\label{def: approx-efficiency}
		A diffusion auction mechanism $M$ is {\em $(\epsilon, \delta)$-efficient} if for any $\boldsymbol{t} \in \boldsymbol{T}$, 
		\begin{align*}
			\epsilon \cdot \mathbf{E}[W(\boldsymbol{t}, M)] + \delta \ge \max \limits_{i \in N, t_i \ne \nil} v_i.
		\end{align*}
	\end{definition}
	
	However, as Table \ref{tab:example} shows, the existing IC mechanisms, including IDM, CDM, FDM, NRM, STM, and SCM, are not $(\epsilon, \delta)$-efficient for any $\delta \in (0, 1)$. We give an example to show the inefficiency. In Figure \ref{fig3}, there are only two buyers in the network. Buyer $a$ becomes the winner and does not need to pay under the above mechanisms. In this case, the social welfare is $0$. However, the valuation of $b$ is $1$, so for any $\delta \in (0, 1)$, these mechanisms are not $(\epsilon, \delta)$-efficient.

	By contrast, $f$-PDM can ensure non-trivial approximate efficiency. For the case in Figure \ref{fig3}, buyer $b$ wins the item with probability 1 with a charge of $0.5$, and buyer $a$ gets a reward of $0.5$. Compared to other mechanisms, all the buyers have higher utility and the social welfare changes from $0$ to $1$. The guarantee of approximate efficiency is formalized as follows.
	
	\begin{theorem} \label{appeff}
		For any incentive diffusion map $f$ and any $\delta \in (0, 1)$, $f$-PDM is $(\frac{1}{2 \delta}, \delta)$-efficient.
	\end{theorem}
	
	\begin{proof}
		Without loss of generality, suppose $f(G) = (1, 2, \ldots , n)$. Since $f$-PDM is IC and the buyer can be ignored if her bid is lower than the highest bid among the buyers before her, we can assume $v_1 \le v_2 \le \cdots \le v_n$.  The expected social welfare
		\begin{align*}
			\mathbf{E}[W(\boldsymbol{t}, M)] = &\; \sum_{i=1}^n v_i \cdot \pi_i(\boldsymbol{t}) 
			= \; (1-v_n+v_1) \cdot v_1 + \sum_{i=2}^n v_i \cdot (v_i - v_{i-1}) \\
			\ge &\; v_1 \cdot v_1 + \sum_{i=2}^n \frac{v^2_i - v^2_{i-1}}{2} \ge \frac{v^2_n}{2}.
		\end{align*}
		Then $\frac{1}{2 \delta} \cdot \mathbf{E}[W(\boldsymbol{t}', M)] + \delta \ge v_n.$ Therefore, $f$-PDM is $(\frac{1}{2 \delta}, \delta)$-efficient. \qed
	\end{proof}
	
	\begin{remark} %[The limitation of $f$-PDM]
		% $(\epsilon,\delta)$-approximate efficiency意味着除了极少数情况，机制都能够保证近似效用。例如，假设每个人的价值分布是iid的，并且存在常数c，$\Pr[v \le c] \le 1 $，那么对于有n个买家的扩散拍卖，无论社会网络的结构如何，PDM有$1-p^n$的概率满足$(\frac{2}{c}, 0)$-approximate efficiency.
		$(\epsilon,\delta)$-approximate efficiency implies that the mechanism guarantees approximate efficiency in all but a negligible fraction of cases. For PDM, as shown in the proof of Theorem \ref{appeff}, its expected welfare $\mathbf{E}[W(\boldsymbol{t}, M)] \ge \frac{v^2_n}{2}$, where $v_n$ is the highest bid. This implies $(\frac{2}{v_n}, 0)$-approximate efficiency. As long as the highest valuation $v_n$ is not arbitrarily close to zero (e.g., $v_n \ge c >	0$), our mechanism achieves a constant-factor multiplicative approximation. For example, assuming that each individual's value distribution is i.i.d., and there exists a constant $c > 0$ such that $\Pr[v \le c] = p < 1 $, then for a diffusion auction with $n$ buyers, regardless of the social network structure, the PDM achieves $(\frac{2}{c}, 0)$-approximate efficiency with probability $1-p^n$.	
		
		%Theorem \ref{appeff} shows that $f$-PDM ensures $(\epsilon,\delta)$-approximate efficiency with $\epsilon=\frac{1}{2\delta}$, where the $\delta$ must be non-zero (otherwise the $\epsilon$ goes to infinity). This indicates that the mechanism does not satisfy any $(\epsilon,0)$-efficiency with finite $\epsilon$, which is a limitation of our mechanism.
		
		%However, for any $\epsilon$, $(\epsilon, 0)$-efficiency is hard to achieve for any diffusion auction mechanism. An idea is to normalize the probabilities, i.e., the winning probability depends exclusively on the proportional share of the aggregate bid. However, it leads to IC violation, whose proof is shown in Appendix \ref{AppE0}. 
	\end{remark}

	\subsection{Collusion-proofness}
	
	Previously, SP provides a crucial guarantee against a single agent manipulating the outcome using multiple identities. However, it does not address the more general and practical threat of distinct agents forming a coalition to act in concert. We refer to such a coalition as a {\em cartel}.
	
	To provide a rigorous foundation for our analysis of cartel behavior, we make two reasonable assumptions about its structure and objective:
	
	\begin{enumerate}[label*=\arabic*)]
		\item {\em Connectivity}: The subgraph induced by the cartel members must be connected. This captures the practical requirement that members must be able to communicate and coordinate their strategies.
		\item {\em Unified objective}: The cartel acts as a single strategic entity to maximize its collective payoff. We model this by assuming its members agree on a common effective valuation for the item.
	\end{enumerate}
	
	Building on this model of a cartel, the property of collusion-proofness (CP), which ensures that no such group can benefit from coordinated misreporting, is formally defined as follows.
	
	%写成卡式积的形式
	\begin{definition} \label{CPdef}
		A diffusion auction mechanism $M$ is {\em collusion-proof (CP for short)} if for any cartel $C = \{ i_1, i_2, \ldots, i_k \} \subseteq N, \, (t'_{i_1}, t'_{i_2}, \ldots, t'_{i_k}) \in \boldsymbol{T}_{i_1} \times \boldsymbol{T}_{i_2} \times \cdots \times \boldsymbol{T}_{i_k}$, and $t_{i_j} = (v_{i_j}, r_{i_j})$ is the true type of $i_j \, (j = 1, 2, \ldots, k)$, where $ v_{i_1} = v_{i_2} = \cdots = v_{i_k}, r'_{i_j} \subseteq r_{i_j} \cup C$, and the subgraph $G[C]$ is connected, the following property holds.
		\begin{align*}
			\; \sum_{j=1}^k \mathbf{E}[u_{i_j}(t_{i_1}, t_{i_2}, \ldots, t_{i_k}, \boldsymbol{t}'_{-C}, M)] 
			\ge \; \sum_{j=1}^k \mathbf{E}[u_{i_j}(t'_{i_1}, t'_{i_2}, \ldots, t'_{i_k}, \boldsymbol{t}'_{-C}, M)].
		\end{align*}
		
	\end{definition}
	
	Different from group-strategyproofness, collusion-proofness treats colluding members as a single entity with an identical valuation. It guarantees that no matter how each member acts, the coalition's overall utility cannot increase (though some members' individual utility may decrease) when they truthfully report. In contrast, group-strategyproofness ensures that no member's utility will decrease. Moreover, CP and SP are dual: CP ensures that no coalition of buyers can profit by misreporting or reducing their collective scale, while SP guarantees that no individual buyer can benefit by misreporting or inflating their scale. In other models, the potential collusive buyer sets (called cartels) are common knowledge ~\cite{collusion1,collusion2} or only multiple buyers becoming one buyer is considered ~\cite{collusion3}. However, our definition relaxes the restrictions on cartels, which implies our definition is stronger.
	
	\begin{proposition}
		The CP problem defined by multiple cartels as common knowledge can be reduced to the CP problem in Definition \ref{CPdef}. 
	\end{proposition}
	
	\begin{proof}
		The reduction is as follows. If a mechanism $M$ is CP, given any $k$ cartels $C_1, \ldots, C_k$, the mechanism $M'$ works as follows.
		
		\begin{enumerate}[label*=\arabic*)]
			\item Set each cartel connected (for example, set $G[C_1], \ldots, G[C_k]$ to be complete graphs).
			\item Set each buyer's bid as the highest bid in her cartel. 
			\item Perform $M$.
		\end{enumerate} 
		
		Since $M$ is CP, each cartel's valuation equals the highest valuation in the cartel and every buyer will report truthfully, i.e., the highest valuation in her cartel. Therefore, being truthful is a dominant strategy under $M'$, since every buyer's bid will be set as the highest bid in her cartel. \qed
	\end{proof}
	
	Note that the cartel can be arbitrary and they can misreport or even do not participate to manipulate the auction and gain more benefits, CP is hard to achieve. In fact, no diffusion mechanism has achieved CP so far. For example, Figure \ref{fig4} shows a collusion case. If no one colludes, buyer $c$ wins and pays $0.1$, others get no reward under IDM or CDM. However, if $a$ and $b$ collude and one of them does not participate in the auction, they will win the item for free. Therefore, diffusion auctions are at risk of collusion.
	
	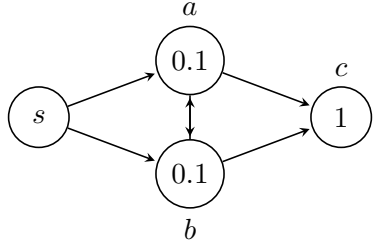
\begin{figure}[t!]
		\begin{center}
			\begin{tikzpicture}[->,>=stealth,shorten >=1pt,auto, semithick]
				
				\node at (0,0) [circle, draw, minimum width=0.8cm, minimum height=0.8cm] (s) {$s$};
				\node at (2, 0.75) [circle, draw, label=above:$a$, minimum width=0.8cm, minimum height=0.8cm] (a)  {$0.1$};
				\node at (2, -0.75) [circle, draw, label=below:$b$, minimum width=0.8cm, minimum height=0.8cm] (b) {$0.1$};
				\node at (4, 0) [circle, draw, label=above:$c$, minimum width=0.8cm, minimum height=0.8cm] (c) {$1$};
				
				\path (s) edge node {} (a)
				(s) edge node {} (b)
				(a) edge node {} (b)
				(a) edge node {} (c)
				(b) edge node {} (a)
				(b) edge node {} (c);
			\end{tikzpicture}
			\caption{Example of collusion. The seller $s$ knows buyer $a$ and buyer $b$, both of whom have valuation $0.1$ and know buyer $c$ and each other. Buyer $c$, whose valuation is $1$, does not know anyone. If $a$ and $b$ collude, they can win the item and have $0.1$ utility.}
			\label{fig4}
		\end{center}
	\end{figure}
	
	However, $f$-PDM is collusion-proof after making slight changes.
	
	\begin{theorem} \label{collu}
		If $f$ is the breadth-first map and the first buyer in $f(G)$, denoted by $i$, pays extra $\frac{1}{2}(v^*_{N_{-C_i}})^2$ rather than $\frac{1}{2}(v^*_{N_{-i}})^2$, $f$-PDM is CP, where $N_{-C_i}$ is the node set of connected components without $i$ after $s$ being removed, i.e., the set of buyers that cannot collude with $i$.
	\end{theorem}
	
	\begin{proof}
		First of all, we prove PDM is CP. Given any cartel $C = \{i_1, i_2, \ldots, i_k\} \subseteq N$ and $i_1 \le i_2 \le \cdots \le i_k$, Sybil-proofness shows that multiple buyers' utility is not higher than one buyer's. Then the cartel's utility $u_C$ is not higher than $u_{i_1} = u_{i_1}(t_{i_1}, \boldsymbol{t}'_{-i_1}, M) $. When $i_1, i_2, \ldots, i_k$ report truthfully, their utility $u_C$ achieves the maximum value $u_{i_1}$. Therefore, PDM is collusion-proof.
		
		From the above proof, we know that the cartel's utility is not higher than its foremost member's utility after the map. Let $l(\boldsymbol{t})$ be the shortest distance from $s$ to the cartel, and $a(\boldsymbol{t}), b(\boldsymbol{t})$ be the number of cartel members and non-cartel members whose distance from $s$ is not longer than $l(\boldsymbol{t})$. When the cartel members misreport $\mathbf{t'}$, $l(\boldsymbol{t}) = l(\mathbf{t'}), a(\boldsymbol{t}) = a(\mathbf{t'}), b(\boldsymbol{t}) = b(\mathbf{t'})$ since $f$ is the breadth-first map. Therefore, the utility of the cartel's foremost member cannot be higher when they misreport, and thus $f$-PDM is CP. \qed
	\end{proof}
	
	In the case of Figure \ref{fig4}, the breadth-first map $f$ will map the graph to $(a, b, c)$ or $(b, a, c)$, and $a$ or $b$ does not need to pay additionally. Therefore, in this case, CP is satisfied under $f$-PDM.

	\subsection{The Seller's Revenue}
	In the last part of this section, we discuss the seller's revenue, which is also the utility of the seller. Similar to efficiency, we define the approximate version of the seller's revenue.
	
	\begin{definition}
		A diffusion auction mechanism $M$ has {\em $(\epsilon, \delta)$-approximate revenue} if for any $\boldsymbol{t} \in \boldsymbol{T}$ and any IC diffusion mechanism $M'$,
		\begin{align*}
			\epsilon \cdot \mathbf{E}[u_s(\boldsymbol{t}, M)] + \delta \ge \mathbf{E}[u_s(\boldsymbol{t}, M')].
		\end{align*}
	\end{definition}
	
	In the definition of approximate efficiency, the right side of the inequality is $\max \nolimits_{i \in N, t_i \ne \nil} v_i$, because $\mathbf{E}[W(\boldsymbol{t}, VCG)]$ equals to the highest bid due to the efficiency of VCG ~\cite{idm}. However, the diffusion mechanism that can maximize the seller's revenue is unknown, so the right side of the inequality cannot be simplified.
	
	When the network is a path graph, $f$-PDM becomes to PDM and the seller's revenue is 0, which is the same as other mechanism. In other cases, the seller's revenue of $f$-PDM is $\frac{1}{2}(v^*_{N_{-i}})^2$, where $N_{-i}$ is the set of buyers who can participate in the auction when $i$ is not invited. As long as $i$ is not the diffusion critical node of the buyer with the highest bid, $u_s = \frac{1}{2}(v^*_{N})^2$. Therefore, the approximate revenue of $f$-PDM is formalized as follows.
	
	\begin{theorem} \label{apprev}
		For any $\delta \in (0, 1)$, $f$-PDM has $(\frac{1}{2 \delta}, \delta)$-approximate revenue when the first buyer after the map is not the diffusion critical node of the buyer with the highest bid.
	\end{theorem}
	
	\begin{proof}
		
		When the first buyer after the map is not the diffusion critical node of the buyer with the highest bid, $u_s(\boldsymbol{t}, f\text{-PDM}) = \frac{1}{2}(v^*_{N})^2$.
		% for any $\boldsymbol{t} \in \boldsymbol{T}$. 
		Then for any mechanism $M$, its revenue $u_s(\boldsymbol{t}, M) \le v^*_N$ because $M$ is IC. Therefore, $\frac{1}{2 \delta} \cdot u_s(\boldsymbol{t}, f\text{-PDM}) + \delta \ge \frac{1}{2 \delta} \cdot \frac{1}{2}(v^*_{N})^2 + \delta \ge v^*_N \ge u_s(\boldsymbol{t}, M).$ Therefore, $(\frac{1}{2 \delta}, \delta)$-approximate revenue is satisfied. \qed
	\end{proof}
	
	As a matter of fact, this amount of revenue can be achieved with high probability, since there is only one diffusion critical node of the buyer with the highest bid can be mapped to the first buyer of the path graph. Therefore, we have the following corollary.

	\begin{corollary}
		If $f$ is the breadth-first map and the seller has $k > 1$ neighbors, $f$-PDM has $(\frac{k}{2(k-1) \delta}, \delta)$-approximate revenue for any $\delta \in (0, 1)$.
	\end{corollary}
	
	\iffalse
	\begin{proof}
		
		Let $p$ be the probability that the first buyer after the map is the diffusion critical node of the buyer with the highest bid. When $f$ is the breadth-first map, $p \le \frac{1}{k}$, since $s$ has $k$ neighbors. Then the expected revenue $ \mathbf{E}[u_s(\boldsymbol{t}, f\text{-PDM})] = \frac{k-1}{2k}(v^*_{N})^2$ for any $\boldsymbol{t} \in \boldsymbol{T}$. Then for any mechanism $M$, $\frac{k}{2(k-1) \delta} \cdot \mathbf{E}[u_s(\boldsymbol{t}, f\text{-PDM})] + \delta \ge \frac{k}{2(k-1)\delta} \cdot \frac{k-1}{2k}(v^*_{N})^2 + \delta \ge v^*_N \ge u_s(\boldsymbol{t}, M).$ Therefore, $f$-PDM has $(\frac{k}{2(k-1) \delta}, \delta)$-approximate revenue. \qed
	\end{proof}
	\fi
	%Moreover, the seller's revenue $\frac{1}{2}(v^*_{N_{-i}})^2$ is well-designed. Because to guarantee IR and IC, the extra charge for the first buyer needs to satisfy that her expected utility is higher than the utility of being the second buyer, and $\frac{1}{2}(v^*_{N_{-i}})^2$ is a quite tight value.
	
	\section{Mechanisms on Multi-Unit Diffusion Auctions} \label{sec5}
	The multi-unit setting emerges as the natural generalization of single-unit auction theory. In this section, we generalize our mechanism to accommodate single-demand multi-unit diffusion auctions, denoted as MUPDM, which satisfies IC and achieves approximate efficiency. We further design the Sybil-proof version, named SP-MUPDM.
	
	\subsection{Multi-Unit PDM}
	
	In single-demand multi-unit auction settings, there are $m$ items for sale ($m \geqslant 1$), and each buyer demands exactly one item.
	A multi-unit diffusion auction mechanism allocates the items to the buyers and determines the corresponding payments. Additionally, properties such as IR, IC, and WBB in multi-unit settings are the same as those in single-unit auctions. However, designing IC multi-unit diffusion auction mechanisms presents significant challenges. Specifically, a naive approach of executing the $f$-PDM mechanism $m$ times does not preserve IC, as it may incentivize the first bidder to strategically lower their bid to avoid selection, thereby increasing their potential profit. For example, if $m=2$ and the network is identical to that shown in Figure \ref{fig1}, where $v_a=v_c=1, v_b=0$. The buyer $a$ can strategically report her valuation $\frac{1}{2}$ rather than $1$ and her utility increases from $\frac{3}{4}$ to $\frac{31}{32}$ when $f$ is the breadth-first map.
	
	\iffalse
	\begin{algorithm}[h]
		\caption{MUPDM}
		\label{alg1}
		\KwIn{A graph $G=(N, E)$, number of items $m$}
		\KwOut{The results of allocation and payment}
		Map the graph $G$ to the path graph $P$ drawn from $f(G)$, where $f$ is the breadth-first map \; 
		Assuming $P=(1, 2, \cdots, n)$ \; 
		\If{$n \le m$}{
			\Return{Everyone gets an item for free}\; 
		}
		\For{$i=1$ to $m$}{
			\tcp{Put the $i$-th buyer to the $i$-th list.}
			$P_i \gets [i] $\;
		}
		\For{$i=m+1$ to $n$}{
			Select an element $j$ in $\{1, 2, \cdots, m\}$ with equal probability \; 
			$P_j.append(i)$\;
		}
		\For{$i=1$ to $m$}{
			Perform PDM on $P_i$ \; 
			The first buyer $j$ in $P_i$ pays extra $\frac{1}{2}(v^*_{{P_i}_{-j}})^2$, where ${P_i}_{-j} = P_i  ~ \backslash ~ \{ ~ l ~ | ~ j \preccurlyeq l ~ \}$\;
		}
		\Return {The results of allocation and payment from above} \;
	\end{algorithm}
	\fi
	
	Therefore, we design a novel mechanism named MUPDM, which is shown as follows.
	
	\begin{definition} [MUPDM] \label{defMU}
		Given a graph $G$, multi-unit probabilistic diffusion mechanism (MUPDM for short) works as follows.
		
		\begin{enumerate}[label*=\arabic*)]
			\item Map the graph $G$ to the path graph drawn from $f(G)$, where $f$ is the breadth-first map.
			\item For the first $k = \min \{ m, |r_s|\}$ buyers, let the $i$-th buyer ( $1 \le i \le k$) be the first buyer of the $i$-th path graph $P_i$.
			\item For each remaining buyer in sequence, select a path graph uniformly at random and append her to the terminal node.
			\item For each path graph $P_i$, perform PDM on $P_i$ and the first buyer $j$ in $P_i$ pays extra $\frac{1}{2}(v^*_{{P_i}_{-j}})^2$, where ${P_i}_{-j} = P_i  ~ \backslash ~ \{ ~ l ~ | ~ j \preccurlyeq l ~ \}$.
		\end{enumerate}
		
	\end{definition}
	
	%\paragraph{Remark.} 
	Intuitively, MUPDM allocates $k = \min \{ m, |r_s|\}$ items through a two-phase process: First, it projects the network $G$ onto $k$ vertex-disjoint path graphs $\{ P_1, \ldots, P_k \}$, where each $P_i$ originates from a distinct neighbor of the seller (formally, $src(P_i) \in r_s$). Second, it executes PDM independently on each $P_i$. 
	We will use the example in Figure \ref{fig1} to facilitate a better understanding of our mechanism. Now there are $m=2$ items. In the first step, $G$ is mapped to the path graphs corresponding to Case 1 and Case 3 in Figure \ref{fig2}, each with a probability of $0.5$. Subsequently, in both scenarios, buyers $a$ and $b$ occupy the first position in one path graph each. In the third step, buyer $c$ is equally likely to be assigned to the path graph of either buyer $a$ or $b$, as illustrated in Figure \ref{figMUPDM}. Finally, the PDM is executed on each path graph. For Case 1, the winning probabilities for $a$ and $c$ are 0.4 and 0.6 respectively. The expected social welfare is $0.66$ and the seller's revenue is $0$. In case 2, the expected social welfare is $0.3 + 0.9 \times 0.9 = 1.11$ and the seller's revenue is $0.405$. By comparison, the expected social welfare of MUDAN is $0.3 + 0.9 \times 0.5 = 0.75$ and the seller's revenue of MUDAN is $0$.
	%我们同样通过图1的例子来更好的理解我们的机制。现在假设有$m=2$个物品。在执行第一步时，$G$以各0.5的概率映射到图2中case 1和case 3对应的path graph。接下来，无论哪种情况，买家$a$和买家$b$各占据一个path graph的第一个买家。在第三步时，买家$c$等可能地落在买家$a$或买家$b$的path graph中，如图7所示。最后，在每个path graph上执行PDM。对于情况1，
	
	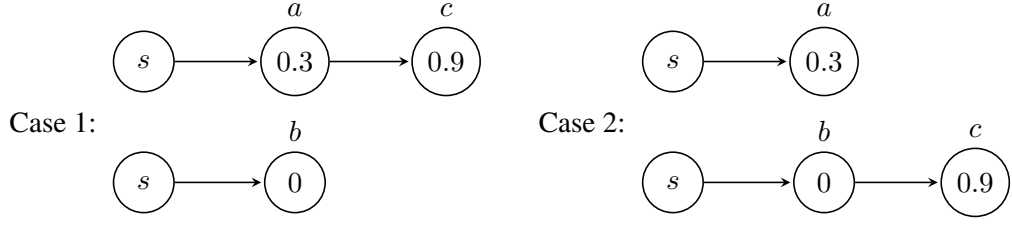
\begin{figure}[t!]
		\begin{center}
			\begin{tikzpicture}[->,>=stealth,shorten >=1pt,auto, semithick]
				\node at (-6.2,0) (n2) { Case 1: };
				\node at (-5,0.8) [circle, draw, minimum width=0.8cm, minimum height=0.8cm] (s21) {$s$};
				\node at (-3, 0.8) [circle, draw, label=above:$a$, minimum width=0.8cm, minimum height=0.8cm] (a2)  {$0.3$};
				\node at (-5,-0.8) [circle, draw, minimum width=0.8cm, minimum height=0.8cm] (s22) {$s$};
				\node at (-3, -0.8) [circle, draw, label=above:$b$, minimum width=0.8cm, minimum height=0.8cm] (b2) {$0$};
				\node at (-1, 0.8) [circle, draw, label=above:$c$, minimum width=0.8cm, minimum height=0.8cm] (c2)  {$0.9$};
				\node at (0.8,0) (n1) { Case 2: };
				\node at (2,0.8) [circle, draw, minimum width=0.8cm, minimum height=0.8cm] (s11) {$s$};
				\node at (2,-0.8) [circle, draw, minimum width=0.8cm, minimum height=0.8cm] (s12) {$s$};
				\node at (4, 0.8) [circle, draw, label=above:$a$, minimum width=0.8cm, minimum height=0.8cm] (a1)  {$0.3$};
				\node at (6, -0.8) [circle, draw, label=above:$c$, minimum width=0.8cm, minimum height=0.8cm] (c1) {$0.9$};
				\node at (4, -0.8) [circle, draw, label=above:$b$, minimum width=0.8cm, minimum height=0.8cm] (b1)  {$0$};
				
				\path (s11) edge node {} (a1)
				(b1) edge node {} (c1)
				(s12) edge node {} (b1)
				(s21) edge node {} (a2)
				(a2) edge node {} (c2)
				(s22) edge node {} (b2);
			\end{tikzpicture}
			\caption{Example of the map of MUPFM. There are 2 cases that MUPDM might map to under the network in Figure \ref{fig1}. Then PDM is performed.}
			\label{figMUPDM}
		\end{center}
	\end{figure}

	Now we present the properties of MUPDM.
	
	\begin{theorem} \label{MUPDMThm}
		MUPDM is WBB, IR, and IC.
	\end{theorem}
	
	\begin{proof}
		Since PDM is WBB and the seller does not need to pay the buyers, MUPDM is WBB. The proof for IR is the same as the proof for Lemma \ref{IR}. Now we focus on the proof for IC.
		
		Note that PDM is IC and $P$ and $P_i \, (1 \le i \le m)$ in Definition \ref{defMU} are independent to the bids, we only need to show every buyer will report her neighbors truthfully. 
		
		When $i$ is not the first $m$ buyers, whether the buyer $i$ invites her neighbors or not, her utility is the same. Note that her utility only depends on the buyers in front of her, we only need to prove for any set of other buyers $A$, the probability $q_{A}= \Pr \nolimits_{\sigma \sim f} [ \{ \sigma(j) \mid j \in A \} = \{ 1, 2, \ldots, \sigma(i)-1 \}]$ remains unchanged. The reason is as follows. Suppose $A=\{ a_1, a_2, \ldots, a_{k-1}\},  \, p_1 = \Pr \nolimits_{\sigma \sim f}[\sigma(a_1)=1], \, p_j = \Pr \nolimits_{\sigma \sim f}[\sigma(a_j)=j \mid \sigma(a_1)=1, \ldots, \sigma(a_{j-1})=j-1] (2 \le j \le k)$, where $a_k=i$. Then removing the outgoing edges of $i$ does not change $p_j (1 \le j \le k)$ due to the characteristics of the breadth-first map $f$. Therefore, $q_{A}$ remains unchanged, which indicates not inviting neighbors cannot gain more utility.
		
		When $i$ is the first $m$ buyers (the first $m$ buyers depend solely on $r_s$), her utility increases when inviting neighbors and is no less than the utility when she is not the first $m$ buyers. Therefore, she will report truthfully.
		
		Therefore, MUPDM is WBB, IR, and IC. \qed
	\end{proof}
	
	The definition of approximate efficiency can be extended to multi-unit scenarios, which provides an approximate guarantee of social welfare.
	
	\begin{definition}
		A multi-unit diffusion auction mechanism $M$ is {\em $(\epsilon, \delta)$-efficient} if for any $\boldsymbol{t} \in \boldsymbol{T}$, 
		\begin{align*}
			\epsilon \cdot \mathbf{E}[W(\boldsymbol{t}, M)] + \delta \ge \sum_{i \in A} v_i,
		\end{align*}
		where $A \subseteq N, |A| =\min \, \{m, n\}$ and for any $i \in A, j \in N \backslash A, v_i \ge v_j$. 
	\end{definition}
	
	Similar to $f$-PDM, MUPDM also ensures approximate efficiency.
	
	\begin{theorem} \label{MUPDMeff}
		For any $\delta \in (0, 1)$, MUPDM is $(\frac{e}{2(e-1)\cdot \delta}, m\delta)$-efficient when $|r_s| \ge m$.
	\end{theorem}
	
	\begin{proof}
		First, for any $\delta \in (0, 1)$ and any path graph $P_i$, let $v^*_{P_i}$ be the highest bid in $P_i$, then we have $\frac{1}{2 \delta} \cdot \mathbf{E}[W(\boldsymbol{t}, P_i, PDM)] + \delta \ge \mathbf{E}[v^*_{P_i}]$ by Theorem \ref{appeff}, where $W(\boldsymbol{t}, P_i, PDM)$ is the social welfare of $P_i$ under PDM.
		
		Second, assume $v_1 \ge v_2 \ge \cdots \ge v_n$. Let $v^*=\sum_{i=1}^{m}v_i$, $v^*_{P_i}$ be the highest bid in $P_i$, $A$ be the set of buyers whose bids are top-$m$ highest, and $p_k$ be the probability that the top-$m$ bidders fall into exactly $k$ different path graphs. Without loss of generality, let the path graphs be $P_1, P_2, \ldots, P_k$ and there are $a_i$ bidders that fall into $P_i \, (1 \le i \le k)$. Now we claim $ \mathbf{E}[\sum_{i=1}^k v^*_{P_i}] \ge \frac{k}{m}v^*$.
		
		By symmetry, for any $1 \le i \le k$, 
		\begin{align*}
			\mathbf{E}[v^*_{P_i}] = & \, \frac{1}{\dbinom{m}{a_i}} \sum_{B \subseteq A, |B| = a_i} \max_{j \in B} v_j \ge 
			\frac{1}{\dbinom{m}{a_i}} \sum_{B \subseteq A, |B| = a_i} (\frac{1}{a_i} \sum_{j \in B} v_j) \\
			= & \, \frac{1}{a_i\dbinom{m}{a_i}} \sum_{j \in A} \dbinom{m-1}{a_i-1}v_j = \frac{1}{m} v^*.
		\end{align*}
		Then we have $\frac{1}{2 \delta} \cdot \mathbf{E}[W(\boldsymbol{t}, P_i, PDM)] + \delta \ge \mathbf{E}[v^*_{P_i}] \ge \frac{1}{m} v^*$.
		Furthermore, \begin{align*}
			\mathbf{E}[W(\boldsymbol{t}, MUPDM)] = & \,\sum_{i=1}^m \mathbf{E}[W(\boldsymbol{t}, P_i, PDM)] \\ \ge & \, \sum_{k=1}^m (p_k \sum_{i=1}^k \mathbf{E}[W(\boldsymbol{t}, P_i, PDM)]) \\ \ge & \, \sum_{k=1}^m p_k \cdot k(2 \delta (\frac{1}{m} v^*-\delta) ).        
		\end{align*}
		
		Now we prove $\sum_{k=1}^m p_k \cdot k > (1-\frac{1}{e}) \cdot m$. Note that $\sum_{k=1}^m p_k \cdot k =  \mathbf{E}[\sum_{i=1}^m \mathbf{1}_{A}(P_i)]$, where $$\mathbf{1}_{A}(P_i)= \begin{cases}
			1, & \text{if } P_i \cap A \neq \emptyset \\
			0, & \text{if } P_i \cap A = \emptyset
		\end{cases}.$$
		Then \begin{align*}
			\sum_{k=1}^m p_k \cdot k = & \, \mathbf{E}[\sum_{i=1}^m \mathbf{1}_{A}(P_i)] = \sum_{i=1}^m\mathbf{E}[ \mathbf{1}_{A}(P_i)] =  \sum_{i=1}^m(1-(1-\frac{1}{m})^m) \\ = & \, m(1-(1-\frac{1}{m})^m) > (1-\frac{1}{e}) \cdot m.  
		\end{align*}
		Therefore, \begin{align*}
			& \,\frac{e}{2(e-1) \delta} \cdot \mathbf{E}[W(\boldsymbol{t}, MUPDM)] + m\delta \\
			\ge & \, \frac{e}{2(e-1) \delta} \cdot (\sum_{k=1}^m p_k \cdot k(2 \delta (\frac{1}{m} v^*-\delta) )) + m\delta \\
			> & \, \frac{e}{2(e-1) \delta} \cdot ((1-\frac{1}{e}) \cdot m(2 \delta (\frac{1}{m} v^*-\delta) )) + m\delta \\
			= & \, v^*.
		\end{align*}
		
		Therefore, for any $\delta \in (0, 1)$, MUPDM is $(\frac{e}{2(e-1) \delta}, m\delta)$-efficient. \qed
	\end{proof}
	
	%最后，我们说明MUPDM的时间复杂度。第一步通过宽度优先搜索将所有买家排序，需要O(|E\cup r_s|)时间（由于是连通图，因此|V|=O(|E'\cup r_s|)）。随后每个买家依次被映射到某个路径图中，需要O(n)时间，最后在每个路径图中执行PDM，用时是每个路径图中的节点数目的线性时间，因此也是O(n)的。综上，MUPDM时间复杂度为O(|E\cup r_s|)。
	Finally, we discuss the time complexity of MUPDM. The first step sorts all buyers via breadth‑first search, which requires $O(|E\cup r_s|)$ time (since the graph is connected, $n=O(|E'\cup r_s|)$). Subsequently, each buyer is mapped to a path graph sequentially, taking $O(n)$ time. Finally, executing PDM on each path graph takes time linear in the number of nodes within that path graph, and hence also $O(n)$ in total. Therefore, the overall time complexity of MUPDM is $O(|E\cup r_s|)$.
	
	\subsection{Sybil-Proof Multi-Unit PDM}
	In Section \ref{sec4}, we show that single-unit diffusion auction mechanisms are vulnerable to Sybil-attack. In multi-unit setting, Sybil attacks remain a critical challenge. For example, Figure \ref{fig_mul_sp} shows a Sybil-attack case. If $m=2$ and everyone is honest, buyer $a$ and $c$ will win for free in MUDAN \cite{mul6}. However, if $b$ creates more than one Sybil-identities with valuation $0$, she will win the item for free. Therefore, multi-unit diffusion auctions are at risk of Sybil-attack. %女巫攻击仍然是一个巨大问题
	
	\begin{figure}[t!]
		\begin{center}
			\begin{tikzpicture}[->,>=stealth,shorten >=1pt,auto, semithick]
				
				\node at (0,0) [circle, draw, minimum width=0.8cm, minimum height=0.8cm] (s) {$s$};
				\node at (2, 0.75) [circle, draw, label=above:$a$, minimum width=0.8cm, minimum height=0.8cm] (a)  {$0.1$};
				\node at (2, -0.75) [circle, draw, label=below:$b$, minimum width=0.8cm, minimum height=0.8cm] (b) {$0.1$};
				\node at (4, 0.75) [circle, draw, label=above:$c$, minimum width=0.8cm, minimum height=0.8cm] (c) {$0.1$};
				\node at (6, 0.75) [circle, draw, label=above:$d$, minimum width=0.8cm, minimum height=0.8cm] (d) {$0.1$};
				
				\path (s) edge node {} (a)
				(s) edge node {} (b)
				(a) edge node {} (c)
				(c) edge node {} (d);
			\end{tikzpicture}
			\caption{Example of Sybil-attack under multi-unit setting. In MUDAN, buyer $a$ wins one item for free because she has more neighbors than $b$, and then $c$ wins one item for free for the same reason. If $b$ creates more than one Sybil-identities with valuation $0$, she will win one item and have $0.1$ utility.}
			\label{fig_mul_sp}
		\end{center}
	\end{figure}
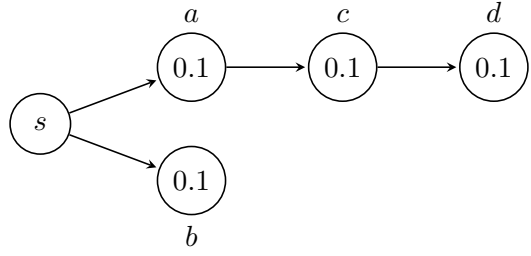
	
	Designing Sybil-proof multi-unit diffusion auction mechanisms is challenging: Beyond creating Sybil identities, a buyer can strategically exclude a neighbor from their competitive circle by deliberately not inviting them. This act makes the neighbor appear as if they are affiliated with an external, non-existent Sybil identity, thereby eliminating them as a direct competitor.
	%买家不仅可以创造女巫身份，还可以通过不邀请邻居使她（邻居）看起来像其他人的女巫身份，将她排除在自己的竞争之外。
	For example, Figure \ref{fig_mul_sp2} shows a case that is challenging to Sybil-proof mechanism design. Note that buyer $e$ is neither a Sybil identity of $a$ nor $b$. However, buyer $d$ can make $e$ appear to be $b$'s Sybil identity by not inviting $e$ and thus reduces competition. %不是买家a或b的女巫身份，但是买家d可以通过不邀请e来让e看起来像b的女巫身份，从而减轻竞争。
	
	\begin{figure}[t!]
		\begin{center}
			\begin{tikzpicture}[->,>=stealth,shorten >=1pt,auto, semithick]
				
				\node at (0,0) [circle, draw, minimum width=0.8cm, minimum height=0.8cm] (s) {$s$};
				\node at (2, 0.75) [circle, draw, label=above:$a$, minimum width=0.8cm, minimum height=0.8cm] (a)  {$0.1$};
				\node at (2, -0.75) [circle, draw, label=below:$b$, minimum width=0.8cm, minimum height=0.8cm] (b) {$0.3$};
				\node at (4, 0.75) [circle, draw, label=above:$c$, minimum width=0.8cm, minimum height=0.8cm] (c) {$0.1$};
				\node at (4, -0.75) [circle, draw, label=below:$e$, minimum width=0.8cm, minimum height=0.8cm] (e) {$0.3$};
				\node at (6, 0.75) [circle, draw, label=above:$d$, minimum width=0.8cm, minimum height=0.8cm] (d) {$0.2$};
				
				\path (s) edge node {} (a)
				(s) edge node {} (b)
				(a) edge node {} (c)
				(b) edge node {} (e)
				(d) edge node {} (e)
				(c) edge node {} (d);
			\end{tikzpicture}
			\caption{Example of social network challenging to Sybil-proof mechanism design.}
			\label{fig_mul_sp2}
		\end{center}
	\end{figure}
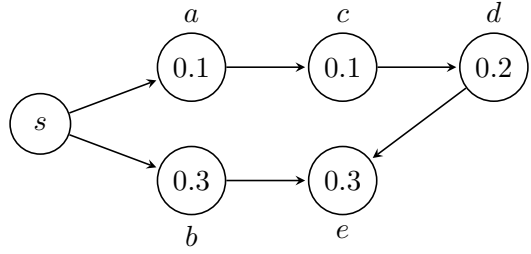
	
	MUPDM is not Sybil-proof because each buyer is randomly assigned to a path graph, enabling buyers to increase their probability of entering their preferred path graph by creating Sybil identities. For example, in Figure \ref{fig_mul_sp2} with $m=2$, buyers $a$ and $b$ become the first buyer in their respective path graphs. Buyer $e$ prefers to be assigned to $a$'s path graph because her utility in $b$'s path graph is $0$. Thus, by creating Sybil identities, $e$ can increase the probability of entering $a$'s path graph, thereby improving her utility. Hence, MUPDM is not SP. 
	
	A natural approach to satisfying Sybil-proofness is to ensure that a buyer and her Sybil identities are allocated to the same path graph, thereby preventing interference with competition in other path graphs. One attempt is to slightly modify MUPDM by requiring that each buyer $i$ is assigned to the same path graph as some buyer $j$ on a shortest path from the seller $s$ to $i$ (where $i$ is a neighbor of $j$), rather than being randomly assigned to a path graph. However, unless the first buyer in each path graph is exempt from payment to the seller, the first buyer has an incentive to avoid inviting others to evade payment. Therefore, we aim to guarantee that a buyer and her Sybil identities are in the same path graph while preserving the property from MUPDM that buyers can be randomly assigned to any path graph.
	%保证SP的一个自然想法是让买家和她的女巫身份在同一个路径图中，这样不会干扰其他路径图中的竞争。一种尝试是对MUPDM稍作修改，让每个买家i和某个从卖家s到i的最短路径上的买家j（i是j的邻居）在同一个路径图中，而非随机进入某个路径图中。然而，除非每个路径图中的首个买家无需向卖家支付，否则首个买家有不邀请其他人从而避免支付的动机。因此，我们希望保证买家和她的女巫身份在同一个路径图中的同时，又具备MUPDM中买家可以随机地进入到每个路径图中的特性。
	
	Based on the above ideas, we first abstract $G$ into a subgraph $G'$ of $G$. Here, $G' = (N, E')$ is a layered graph stratified by the distance of buyers from the seller s, where $G'$ retains only edges between adjacent layers, i.e., $E' = \{ e = (j, i) \mid e \in E \land d(i) = d(j) + 1 \}$ and $d(i)$ denotes the distance from $s$ to $i$. Note that in $G'$, for any node $v$, every path from $s$ to $v$ is a shortest path from $s$ to $v$ in the original graph $G$. Next, we perform a breadth‑first traversal over $G'$, mapping each node to a path graph. Nodes in $G'$ that may be Sybil identities of a given buyer $i$ are mapped to the same path graph as $i$, while nodes not identified as Sybil identities are randomly assigned to any path graph. Based on this, we design the Sybil‑Proof Multi‑Unit PDM (SP‑MUPDM).
	
	%基于以上想法，我们首先将G抽象成G的一个子图G'，G'=(N,E')是一个按买家与卖家s的距离分层的分层图，G'仅保留了相邻层的边，即$E'=\{e=(i, j) \mid d(j) = d(i) + 1 \}$，其中d(i)是从s到i的距离。注意，在G'中，对于任意点v，从s到v的路径都是在图G中从s到v的最短路径。接下来在G'中，进行广度优先遍历，将每个节点映射到每个路径图中。在G'中可能是某个买家i的女巫身份的节点映射到与i相同的路径图中，不是女巫身份的节点随机映射到任意一个路径图中。据此，我们设计了Sybil-Proof Multi-Unit PDM (SP-MUPDM)。
	
	\begin{definition} [SP-MUPDM] \label{defSPMU}
		Given a graph $G$, multi-unit probabilistic diffusion mechanism (MUPDM for short) works as follows.
		
		\begin{enumerate}[label*=\arabic*)]
			\item Map the graph $G$ to the path graphs by Mechanism \ref{algSP}.
			\item For each path graph $P_i$, perform PDM on $P_i$ and the first buyer $j$ in $P_i$ pays extra $p_i = \frac{1}{2}(v^*_{{D_{i,j}}})^2$, where ${D_{i,j}} = P_i  ~ \backslash ~ \{ ~ l ~ | \text{There exists a shortest path from $s$ to $l$ that passes through $j$.} \}$.
		\end{enumerate}
		
	\end{definition}
	
	%与满足SP的单物品扩散拍卖机制PDM和f-PDM不同，SP-MUPDM需要找到那些可能会成为女巫身份的节点。这个过程相当于图论中找到每个节点的支配点，即建立求解有根图(N s,E' rs,s)的支配树，可以通过运行Lengauer-Tarjan Dominators Algorithm或直接求解每个节点的所有前驱节点的最近公共祖先在O(nlogn+m)时间内实现。此外，在支付时，每个路径图的第一个买家i需要额外支付的金额取决于该路径图中除了从s出发某条最短路径经过i的节点中出价最高的节点，这避免了第一个买家通过不邀请其他人的策略行为增加自己的效用。
	Unlike the single-item diffusion auction mechanisms PDM and $f$-PDM that achieve Sybil-proofness, SP-MUPDM requires identifying potential Sybil identity nodes. This process is equivalent to finding the dominator of each node in graph theory \cite{sybil2}, i.e., constructing the dominator tree for the rooted graph $(N \cup \{s\}, E' \cup r_s, s)$. It can be implemented by executing the Lengauer-Tarjan Dominators Algorithm \cite{LTAlg1,LTAlg2} or directly computing the lowest common ancestor of all predecessor nodes for each node, with a time complexity of $O(n \log n + |E'\cup r_s|)$. 
	Furthermore, during payment settlement, the first buyer $i$ in each path graph is required to pay an additional amount determined by the highest bid among nodes in that path graph, excluding those lying on any shortest path from $s$ that passes through $i$. This design prevents the first buyer from increasing their own utility through strategic non-invitation behavior.
	
	We now present a running example of SP-MUPDM for $m = 2$ based on Figure \ref{fig_mul_sp2}. Frist, the edge $(d, e)$ is removed, resulting in the subgraph $G'$. The social network is then mapped into two path graphs as shown in Figure \ref{fig_MUPDM}, since buyer $c$ and $d$ must be in the same path graph as $a$, and buyer $e$ must be in the same path graph as $b$. Finally, the PDM is executed on each path graph. The expected social welfare is $ (0.1 \times 0.9 +  0.2 \times 0.1) + (0.3 \times 1) = 0.41$, and the expected revenue is $0$.
	%接下来我们给出当$m=2$时基于图8的SP-MUPDM的运行实例。首先，边(d, e)被删除，得到新的图G'。接下来社会网络被映射为图9所示的两个path graph，因为买家$c$和$d$必须与$a$在同一个path graph中、买家$e$必须与$b$在同一个path graph中。最后，在每个path graph上运行PDM。此时，期望社会福利是$ (0.1 \times 0.9 +  0.2 \times 0.1) + (0.1 \times 0.8 +  0.3 \times 0.2) = 0.25$，期望收入为$0$。
	
	\begin{figure}[t!]
		\begin{center}
			\begin{tikzpicture}[->,>=stealth,shorten >=1pt,auto, semithick]
				\node at (0,1) [circle, draw, minimum width=0.8cm, minimum height=0.8cm] (s1) {$s$};
				\node at (2, 1) [circle, draw, label=above:$a$, minimum width=0.8cm, minimum height=0.8cm] (a)  {$0.1$};
				\node at (0,-1) [circle, draw, minimum width=0.8cm, minimum height=0.8cm] (s2) {$s$};
				\node at (2, -1) [circle, draw, label=above:$b$, minimum width=0.8cm, minimum height=0.8cm] (b) {$0.3$};
				\node at (4, 1) [circle, draw, label=above:$c$, minimum width=0.8cm, minimum height=0.8cm] (c)  {$0.1$};
				\node at (6, 1) [circle, draw, label=above:$d$, minimum width=0.8cm, minimum height=0.8cm] (d)  {$0.2$};
				\node at (4, -1) [circle, draw, label=above:$e$, minimum width=0.8cm, minimum height=0.8cm] (e)  {$0.3$};
				
				\path (s1) edge node {} (a)
				(a) edge node {} (c)
				(c) edge node {} (d)
				(s2) edge node {} (b)
				(b) edge node {} (e);
			\end{tikzpicture}
			\caption{Example of the map of MUPDM under Figure \ref{fig_mul_sp2}. To prevent Sybil attack, the result of the map is unique.}
			\label{fig_MUPDM}
		\end{center}
	\end{figure}
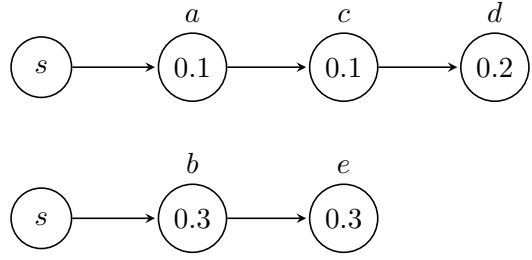
	
	The following theorem states the properties of SP-MUPDM.
	
	\begin{algorithm}
		
	\caption{The map of SP-MUPDM}
	\label{algSP}
	\KwIn {A graph $G=(N, E)$, number of items $m$, neighbors of the seller $r_s$, buyers' bid $v'_1, \ldots, v'_n$.}
	\KwOut{$k = \min \{m, |r_s|\}$ permutations $P_1, \ldots, P_k$, which form a partition of $N$ and payments $p_1, \ldots, p_k$ of the first buyer in each path graph.} %k个排列
	%Initialize $A_{n \times n}$ as $\mathbf{0}$ \tcp*{$A[i][j]$ denotes whether the shortest path from $s$ to $i$ passes through $j$.} %l 是i被分配到的path graph
	
	$d_1, \ldots, d_n \gets 0$ \tcp*{$d_i$ is the distance from $s$ to $i$.}
	
	$a_1, \ldots, a_n \gets 0$ \tcp*{$a_i$ is the path graph to which buyer $i$ is assigned.}
	
	$k = \min \{m, |r_s|\}, f_1, \ldots, f_k \gets 0$ \tcp*{$f_i$ is the first buyer of $i$th path graph.}
	
	$r'_1, \ldots, r'_n \gets \emptyset$ \tcp*{$r'_i$ is the neighbors of $i$ in $G'$.}
	
	\For{$i \in r_s$} {
		$d_i \gets 1$ \; 
	}
	
	Let $q$ be a queue and enqueue the elements from $r_s$ into $q$ \; 
	
	\While{$q$ is not empty \tcp{Compute $G'=(N,E')$.}}{
		$j \gets dequeue(q)$ \; 
		\For{$i \in r_j$}{
			\If{$d_i = 0$ \tcp{$i$ is not visited before.}} 
			{
				$d_i \gets d_j + 1, r'_j \gets r'_j \cup \{i\}$ \; 
				$q.enqueue(i)$ \; 
			}
			\ElseIf{$d_i = d_j + 1$}
			{
				$r'_j \gets r'_j \cup \{i\}$ \; 
			}
		}
	}
	
	Compute the dominator tree of the rooted graph $(N \cup \{s\}, E' \cup r_s, s)$ and obtain the parent node ${dom}_i$ of each node $i$ in the dominator tree.
	
	Initialize $P_1, \ldots, P_k$ as empty vectors \; 
	
	%\tcp{Compute $P_1, \ldots, P_k$.}
	
	Let $P$ be the permutation of $1, 2, \ldots, n$ drawn from $f(G)$, where $f$ is the breadth-first map \; 
	
	\For{$i=1$ to $k$} {
		$a_{P[i]} \gets i, f_i \gets P[i]$ \; 
		
		$P_i.push\_back(P[i])$ \tcp*{$P[i]$ is the first buyer of $i$th path graph.}
	}
	
	\For{$i=k+1$ to $n$} {
		\If{${dom}_{P[i]} = s$ \tcp{$P[i]$ is not a Sybil identity of any other node.}}{
			Randomly select an element $j$ from $1, \ldots, k$ with equal probability \;
		}
		\Else{
			$j = a_{{dom}_{P[i]}}$ %\tcp*{$P[i]$ must be mapped to the same path graph as its dominator node.}
		}
		
		$P_{j}.push\_back(P[i])$ \; 
		
		$a_{P[i]} \gets j$ \; 
	}
	
	\For{$i=1$ to $k$ \tcp{Compute $p_1, \ldots, p_k$.}} {
		$p_i \gets 0, A_{1 \times n} \gets \mathbf{0}$ \tcp*{$A$ is the reachability array with respect to $i$ in $G'$.}
		
		Compute $A$ by DFS or BFS \; 
		
		\For{$j \in P_i$}{
			\If{$A[j] = 0$ and $p_i < v'_j$}{
				$p_i \gets v'_j$ \tcp*{Find the highest value in $P_i \backslash A $.}
			}
		}
		
		$p_i \gets p^2_i/2$ \; 
	}
	
	\Return{$P_1, \ldots, P_k, p_1, \ldots, p_k$}
	\end{algorithm}

	\begin{theorem} \label{SPMUPDMThm}
		SP-MUPDM is WBB, IR, IC, and SP.
	\end{theorem}
	
	\begin{proof}
		The proofs for WBB and IR are similar to the proof for Theorem \ref{MUPDMThm}. Since SP implies IC, we only need to show the mechanism is SP.
		
		%我们分以下两种情况讨论。首先，对于每个path graph的第一个买家，我们考虑每种可能的策略行为。(1) 不邀请某些邻居。
		Our proof is divided into two steps. First, consider the first buyer $i$ in each path graph. The following two aspects demonstrate why the first buyer in each path graph will report truthfully.

		\begin{enumerate}[label=(\arabic*)]
			\item \textbf{Strategic non-invitation and Sybil identity creation.} First, any Sybil identity of $i$ must be placed in the same path graph as $i$, because every path from $s$ to a Sybil identity of $i$ necessarily passes through $i$.
			
			Let $u$ and $u$ denote $i$'s utility under truthful and untruthful reporting, respectively. Let the random variables $X$ and $X'$ be $i$'s expected utility from the PDM execution on the path graph containing $i$ under truthful and untruthful reporting (by Theorem \ref{Thm1}, they equal $v_i+(v'_m-v_i)^2/2$). 
			Let $Y$ and $Y'$ be the corresponding additional payment $p_i$ under truthful and untruthful reporting. Then 
			$$\mathbf{E}[u]=\mathbf{E}[X-Y]=\mathbf{E}[X]-\mathbf{E}[Y], \, \mathbf{E}[u']=\mathbf{E}[X'-Y']=\mathbf{E}[X']-\mathbf{E}[Y'].$$
	We prove $\mathbf{E}[u] \ge \mathbf{E}[u']$ by showing $\mathbf{E}[X] \ge \mathbf{E}[X']$ and $\mathbf{E}[Y] \le \mathbf{E}[Y']$.
			\begin{enumerate}[label=(\arabic{enumi}.\arabic*)]
				\item \textbf{Proof of $\mathbf{E}[X] \ge \mathbf{E}[X']$.} Note that $X$ and $X$ depend only on the highest bid inside the path graph containing $i$, we examine three types of nodes:
				
				\begin{enumerate}[label=(\alph*)] \label{item a}
					\item \label{item:a} \textbf{Nodes previously identified as Sybil identities of $i$ (denoted $j$).} Because $j$ is recognized as a Sybil identity of $i$ when $i$ invites all neighbors, $j$ enters $i$'s path graph with probability $1$. When $i$ creates Sybil identities or withholds invitations, the probability that $j$ enters $i$'s path graph cannot increase. We may therefore assume these nodes still enter $i$'s path graph with probability $1$ (this assumption does not decrease $\mathbf{E}[X']$.
					\item \label{item:b} \textbf{Nodes previously identified as Sybil identities of another node $v$ (denoted $j$).} After $i$ creates Sybil identities or withholds invitations, $j$ remains recognized as a Sybil identity of $v$, because every path from $s$ to $j$ still passes through $v$. Hence $j$ and $v$ are still placed in the same path graph.
					\item \label{item:c} \textbf{Nodes not previously identified as Sybil identities (denoted $j$).} After $i$ creates Sybil identities or withholds invitations, $j$ cannot become recognized as a Sybil identity of $i$. Otherwise, even when $i$ invites all neighbors, every path from $s$ to $j$ would have to pass through $i$, which would already have made $j$ a Sybil identity of $i$ under truthful behavior---a contradiction.
					
					Therefore, after $i$'s deviation, $j$ either remains unidentified as a Sybil identity, or becomes identified as a Sybil identity of some node $v$ that is not a Sybil identity of $i$ (otherwise, by transitivity, $j$ would also be a Sybil identity of $i$). The former has no effect on $\mathbf{E}[X']$, while the latter may lower $\mathbf{E}[X']$.
				\end{enumerate}
				Combining these cases, we conclude $\mathbf{E}[X] \ge \mathbf{E}[X']$.
				\item \textbf{Proof of $\mathbf{E}[Y] \le \mathbf{E}[Y']$.}
				Recall that $p_i = \frac{1}{2}(v^*_{{D_{i,j}}})^2$, where $${D_{i,j}} = P_i  ~ \backslash ~ \{ ~ l ~ | \text{There exists a shortest path from $s$ to $l$ that passes through $j$.} \}. $$ Let $D, D' \subseteq N$ be the sets of nodes lying on some shortest path from $s$ that passes through $i$ under truthful and untruthful reporting, respectively. Then $D' \subseteq D$, because if creating Sybil identities or removing some outgoing edges of $i$ makes a shortest path from $s$ to $j$ pass through $i$, then such a path also exists when those edges are present.
				
				We follow the same three‑type analysis:
				
				\begin{enumerate}[label=(\alph*)]
					\item \textbf{Nodes previously identified as Sybil identities of $i$ (denoted $j$).} Since $ j \in D$, it does not affect $p_i$. We may assume these nodes also belong to $D'$ (this assumption does not increase $\mathbf{E}[Y']$.
					\item \textbf{Nodes previously identified as Sybil identities of another node $v$ (denoted $j$).} As in \ref{item a} \ref{item:b}, $j$ and $v$ remain in the same path graph.
					\item \textbf{Nodes not previously identified as Sybil identities (denoted $j$).} Two sub‑cases arise:
					
					\begin{enumerate}[label=(c.\arabic*)]
						\item $j$ remains unidentified. Then $j$ and its potential Sybil identities have no effect on the probability of entering $i$'s path graph, implying $\mathbf{E}[Y] \le \mathbf{E}[Y']$.
						\item $j$ becomes identified as a Sybil identity of another node $v$. In this case we must have $j \in D$. Indeed, after $i$ creating Sybil identities or removes some outgoing edges, every shortest path from $s$ to $j$ must pass through $v$; before the removal there existed a shortest path that did not pass through $v$. Hence there exists a shortest path (using the removed edges) that passes through $i$, placing $j$ in $D$. Consequently, under truthful reporting, $j$ and its suspected Sybil identities do not influence $p_i$ and thus $\mathbf{E}[Y] \le \mathbf{E}[Y']$.
					\end{enumerate}
				\end{enumerate}
				Combining these cases, we conclude $\mathbf{E}[Y] \ge \mathbf{E}[Y']$.
			\end{enumerate}
			
			From $\mathbf{E}[X] \ge \mathbf{E}[X']$ and $\mathbf{E}[Y] \le \mathbf{E}[Y']$, it follows that $\mathbf{E}[u] \ge \mathbf{E}[u']$. Hence, withholding invitations or creating Sybil identities does not increase buyer $i$'s utility.
			\item \textbf{Misreporting valuation.} Since a bid does not affect the probability of being mapped to a path graph---it only influences the outcome when the PDM is executed---and because truthful reporting maximizes a buyer's utility within the PDM, it follows that truthful reporting also maximizes the buyer's utility in SP‑MUPDM.
		\end{enumerate}
		
		In summary, the first buyer in each path graph will report truthfully.
		
		%接下来我们分以下三点证明路径图中非首位买家i会诚实报告。第一，不邀请邻居j。根据邻居的属性我们分两种情况。(1) d(j) <= d(i)。此时这条边在G'中会被剔除，因此不会对i的效用产生影响。(2) d(j) > d(i)。由于机制通过宽度优先对每个节点依次进行映射，因此i早于j被映射，根据PDM知道，j不会对i的效用产生影响。综上，不邀请邻居j不会使i的效用增加。
		%第二，创造女巫身份。由于i的女巫身份会进入到和i相同的路径图中，并且不会改变比i离s更近的节点与s的距离，因此创造女巫身份不会使i的效用增加。
		%第三，虚报估值。由于估值不会影响映射到路径图的概率，因此虚报估值不会使i的效用增加。
		%综上，诚实报告会使效用最大化。
		Next, we prove that a non‑first buyer $i$ in a path graph will report truthfully by examining the following three aspects.
		
		\begin{enumerate}[label=(\arabic*)]
			\item \textbf{Strategic non-invitation.} We distinguish two cases based on the neighbor $j$’s position.
			\begin{enumerate}[label=(\arabic{enumi}.\arabic*)]
				\item $d(j) \le d(i)$. This edge is removed in the pruned graph $G'$; consequently, it cannot affect $i$'s utility.
				\item $d(j) > d(i)$. Because the mechanism maps nodes via breadth‑first traversal, $i$ is mapped before $j$. By the properties of PDM, $j$ does not influence $i$'s utility.
			\end{enumerate}
			Hence, withholding invitations does not increase $i$'s utility.
			\item \textbf{Creating Sybil identities.} Any Sybil identity of $i$ is placed in the same path graph as $i$, and it does not alter the distances from $s$ to nodes that are closer to $s$ than $i$. Therefore, creating Sybil identities does not increase $i$’s utility.
			\item \textbf{Misreporting valuation.} A buyer’s reported valuation does not affect the probability of being assigned to a particular path graph. Hence, misreporting valuation does not increase $i$'s utility.
		\end{enumerate}
		Since none of the three possible deviations can improve $i$'s utility, truthful reporting is a dominant strategy for any non‑first buyer in a path graph.
		
		Based on the above two points, SP-MUPDM is Sybil-proof. Therefore, MUPDM is WBB, IR, IC, and SP. \qed
	\end{proof}
	
	%在本节的最后，我们说明SP-MUPDM的时间复杂度。在机制3中，5-17行用于计算G'，需要O(|E\cup r_s|)时间（由于是连通图，因此|V|=O(|E'\cup r_s|)）。第18行用于计算支配树，时间复杂度为$O(n \log n + |E'\cup r_s|)$。19-31行用于生成每个路径图，等于图的宽度优先搜索的时间O(|E\cup r_s|)。32-40行用于计算每个路径图中首个买家的支付，需要O(k\dot |E\cup r_s|)时间，因为第35行最坏可能需要O(|E\cup r_s|)时间。综上，SP-MUPDM需要O(n \log n + k\dot |E\cup r_s|)时间，其中$k = \min \{ m, |r_s|\}$。
	At the end of this section, we analyze the time complexity of SP‑MUPDM. In Mechanism \ref{algSP}, lines 5--17 compute $G'$, requiring $O(|E\cup r_s|)$ time (since the graph is connected, $n=O(|E'\cup r_s|)$). Line 18 computes the dominator tree with time complexity $O(n \log n + |E'\cup r_s|)$. Lines 19--31 generate each path graph, which is equivalent to performing a breadth‑first search on the graph and takes $O(|E\cup r_s|)$ time. Lines 32--40 compute the payment of the first buyer in each path graph, requiring $O(k\dot |E\cup r_s|)$ time because line 35 may need $O(|E\cup r_s|)$ time in the worst case. In summary, SP‑MUPDM runs in $O(n \log n + k\dot |E\cup r_s|)$ time, where $k = \min \{ m, |r_s|\}$.
	
	\section{Conclusions and Future Work} \label{sec6}
	%本文首先考虑了在路径图上的机制设计，并提出了PDM。在此基础上，文章将机制推广到一般的网络上并设计了f-PDM。它是IR, IC, WBB的，同时可以保证非平凡的近似社会福利和收入。当f满足适当条件时，f-PDM能够满足SP或CP。此外，我们考虑了将单物品拍卖扩展为多物品拍卖，并在PDM的基础上设计了MUPDM，它满足IR、IC、WBB和近似社会福利。进一步地，我们考虑了多物品拍卖下的女巫攻击问题，并设计出了能够保证SP的SP-MUPDM。
	The paper first proposes the Probabilistic Diffusion Mechanism (PDM) for path graphs and then generalizes it to arbitrary networks via a map $f$, yielding the family of $f$-PDM mechanisms. This family preserves individual rationality (IR), incentive compatibility (IC), and weak budget balance (WBB), while providing constant-factor approximations to both social welfare and revenue. When $f$ satisfies certain conditions (e.g., breadth‑first order), $f$-PDM additionally achieves Sybil‑proofness (SP) or collusion‑proofness (CP). Furthermore, we extend the framework to multi‑item settings and design the Multi‑Unit PDM (MUPDM), which also satisfies IR, IC, WBB and guarantees approximate efficiency. To resist Sybil attacks in multi‑item auctions, we further develop the Sybil‑Proof MUPDM (SP‑MUPDM).
	
	%对于未来工作，我们计划从以下三个方面开展。第一，将机制扩展到有多个卖家或者买家有多个物品的需求的场景。多个卖家和买家需要多个物品是更一般的情况，是否可以以类似的方式设计在这些场景下满足IR、IC、WBB、SP和近似社会福利的机制？第二，设计满足更多性质的机制。本文设计了满足近似社会福利的机制，那么是否存在满足纯近似社会福利的机制呢？除此之外，例如设计满足公平性等性质也是一个有趣的方向。第三，考虑特定网络结构下的机制。在设计对任意社会网络结构都满足特定性质的机制往往是困难的，那么是否可以通过对社会网络的结构加以限制，设计出更简洁、优雅的机制呢？
	For future work, we plan to proceed along the following three directions. First, extending the mechanism to scenarios with multiple sellers or multi-demand multi-unit diffusion auctions. Multiple sellers and buyers requiring multiple items represent a more general setting. Can we design mechanisms in these contexts that satisfy IR, IC, WBB, SP, and approximate social welfare in a similar manner? Second, designing mechanisms that satisfy additional desirable properties. The paper has developed mechanisms achieving $(\frac{1}{2 \delta}, \delta)$-approximate efficiency; a natural follow-up is whether mechanisms achieving $(\epsilon, 0)$-efficiency exist. Furthermore, incorporating other properties such as fairness constitutes another interesting direction. Third, investigating mechanism design under specific network structures. Designing mechanisms that guarantee certain properties for arbitrary social network structures is often challenging. Would it be possible to impose reasonable restrictions on the network topology, leading to simpler and more elegant mechanisms? 
	
\section*{Acknowledgements}
	This work was supported by the National Natural Science Foundation of China under Grant 62572007 and Key Laboratory of Interdisciplinary Research of Computation and Economics (Shanghai University of Finance and Economics), Ministry of Education.

\bibliographystyle{splncs04}
\bibliography{ref}

\end{document}